\documentclass[journal]{IEEEtran}

\usepackage{amsmath}
\usepackage{graphicx}
\usepackage{amssymb}
\usepackage{siunitx}
\usepackage{cite}
\usepackage{mathtools}
\usepackage{bbold}
\usepackage{steinmetz}
\usepackage{stmaryrd}
\usepackage{bbm, dsfont}
\usepackage{amsthm}
\usepackage{verbatim}
\usepackage[version=4]{mhchem}
\usepackage[latin1]{inputenc}
\usepackage{tikz}
\usepackage{dirtytalk}
\usetikzlibrary{shapes,arrows}
\usepackage{enumerate}
\usepackage{color}
\usepackage{pgfplots}
\usetikzlibrary{calc}
\usepackage{mathrsfs}
\usepackage{upgreek}
\usepackage{steinmetz}
\usepackage[outdir=./]{epstopdf}
\usepackage{lipsum}
\usepackage[variablett]{lmodern}
\usepackage{amsbsy}
\usepackage{bm}
\usetikzlibrary{positioning}
\usetikzlibrary{shapes,arrows}
\usepackage{multirow}

\pgfplotsset{
	every tick label/.append style={scale=1},
	every axis/.append style={
	}
}

\pgfplotsset{
	grid style = {
		dash pattern = on 0.05mm off 1mm,
		line cap = round,
		black,
		line width = 0.5pt
	}
}
\usetikzlibrary{shapes.multipart,intersections}

\newcommand{\tx}{\textnormal}

\newsavebox\myboxA
\newsavebox\myboxB
\newlength\mylenA

\newcommand*\xoverline[2][0.75]{%
    \sbox{\myboxA}{$\m@th#2$}%
    \setbox\myboxB\null
    \ht\myboxB=\ht\myboxA%
    \dp\myboxB=\dp\myboxA%
    \wd\myboxB=#1\wd\myboxA
    \sbox\myboxB{$\m@th\overline{\copy\myboxB}$}
    \setlength\mylenA{\the\wd\myboxA}
    \addtolength\mylenA{-\the\wd\myboxB}%
    \ifdim\wd\myboxB<\wd\myboxA%
       \rlap{\hskip 0.5\mylenA\usebox\myboxB}{\usebox\myboxA}%
    \else
        \hskip -0.5\mylenA\rlap{\usebox\myboxA}{\hskip 0.5\mylenA\usebox\myboxB}%
    \fi}
\makeatother
\DeclareMathOperator{\argmin}{argmin}
\DeclareMathOperator{\argmax}{argmax}
\newtheorem{lemma}{Lemma}

\usepackage{subfigure} 
\usepackage{graphicx}

\begin{document}

\title{Computationally Efficient Algorithm for Eco-Driving over Long Look-Ahead Horizons}

\author{Ahad~Hamednia,
        Nalin~Kumar~Sharma,
        Nikolce~Murgovski,
        and~Jonas~Fredriksson 
\thanks{A. Hamednia, N. K. Sharma, N. Murgovski and J. Fredriksson are with the Department
of Electrical Engineering, Chalmers University of Technology,
Gothenburg 412 96, Sweden (e-mail: hamednia@chalmers.se).}}
\maketitle

\begin{abstract}
This paper presents a computationally efficient algorithm for eco-driving over long prediction horizons. The eco-driving problem is formulated as a bi-level program, where the bottom level is solved offline, pre-optimising gear as a function of longitudinal velocity and acceleration. The top level is solved online, optimising a nonlinear dynamic program with travel time, kinetic energy and acceleration as state variables. To further reduce computational effort, the travel time is adjoined to the objective by applying necessary Pontryagin's Maximum Principle conditions, and the nonlinear program is solved using real-time iteration sequential quadratic programming scheme in a model predictive control framework. Compared to standard cruise control, the energy savings of using the proposed algorithm is up to $15.71\,\%$.
\end{abstract}


\IEEEpeerreviewmaketitle

\section{Introduction}\label{sec:intro}

\IEEEPARstart{E}{xcessive} energy consumption of vehicles is recently being regarded as a crucial concern for policy makers and customers due to economic, ecological and environmental issues. For instance, Organisation for Economic Co-operation and Development (OECD) forecasts a rapid growth in transport demand over the coming years, which may lead to 60\% increase in worldwide transport \ce{CO2} emissions by 2050, \cite{oecditf19}. One effective way to mitigate destructive consequences from ever growing energy consumption by vehicles is to improve the vehicular energy efficiency. It should be noted that increased efficiency is the biggest contributor to abatement of the \ce{CO2} emissions in the energy sector according to \cite{oecdiea18}.

Eco-driving has been concerned widely as an approach for energy-efficient manoeuvring of a vehicle by optimising velocity profile when considering road information and traffic flow \cite{kamal11,kamal12,vajedi15,luo15}. When driving in a hilly terrain, it is preferable to vary the vehicle speed over a narrow interval while keeping the maximum allowed travel time, i.e., speeding up when driving downhill and decreasing speed when climbing uphill, to have less energy waste at braking pads compared to a constant speed driving \cite{barkenbus2010eco}. Implementing this behaviour over complex road topographies is generally achieved by model-based control methods that maximise energy efficiency by optimally coordinate the energy use.

Dynamic programming (DP) \cite{bellman57} is the most commonly used algorithm to optimise the velocity profile of 
vehicles due to its potential to tackle non-convex, nonlinear and mixed-integer optimisation problems \cite{hellstrom09,hellstrom10,dib11,heppeler16a,wahl13,buhler13}. Fuel-optimal look-ahead control strategies have been proposed in \cite{hellstrom09} and \cite{hellstrom10} using DP, where in addition to optimising velocity, optimal gear shifting of conventional trucks is also investigated. Furthermore, a DP-based method is applied in \cite{dib11} to minimise the energy consumption in fully electric vehicles (EVs) by optimising vehicle speed on short-range trips, e.g. driving between two consecutive traffic lights. A combined energy management and eco-driving approach using discrete DP is devised in \cite{heppeler16a} for hybrid electric vehicles (HEVs) driving over limited horizons, where the velocity profile is allowed to be optimised to further enhance fuel efficiency. Despite the promising contributions in solving optimal control problems, DP-based methods suffer from the \textit{curse of dimensionality}, which denotes to a fact that computational time increases exponentially with the number of state variables and control signals, \cite{bellman57}. Several ways have been taken to decrease computational effort, for example by limiting the look-ahead horizon of cruise controllers for HEVs. At the current state, real-time capable DP-based control can only be applied for short prediction horizon scenarios of HEVs \cite{wahl13}. Other approaches focus on simplifying the powertrain model, by e.g. using a simplified internal combustion engine (ICE) model or discarding system states, such as travel time, ICE on/off and gear \cite{buhler13}. 

For high-dimensional optimisation problems, e.g. optimal control of HEVs with more energy states, several alternative approaches have been proposed. In \cite{themann15} a mixed-integer quadratic program (MIQP) \cite{boyd04} has been applied for power allocation of HEVs. A way to diminish computational complexity of the high-dimensional problems is adjoining system dynamics to the cost function and neglecting constraints on state variables, as shown in  \cite{hellstrom10b,keulen10,keulen11}. In \cite{keulen11} Pontryagin's Maximum Principle (PMP) \cite{pontryagin62} has been applied to optimise vehicle speed, gear selection and energy use of HEVs, where integer state variables have been neglected. Furthermore, in \cite{held18} minimisation of energy consumption using PMP and considering varying speed requirements has been studied. Although PMP-based methods are computationally efficient for optimal velocity problems over long look-ahead horizons, they do not provide the same computational advantage for problems where state variables often activate their bounds. This is especially relevant for single shooting methods used for solving two-point boundary value problems (2PBVPs), as in e.g. \cite{Keulen14}. 

Another portion of the conducted research benefits from the combination of DP and other methods. Such approaches have been proposed by \cite{murgovski16,johannesson15a,johannesson15b,hovgard18}, where real-valued decisions, e.g., planing optimal velocity, are made by sequential convex optimisation, while integer decisions are taken by DP. These strategies have also been shown to be effective when considering surrounding traffic \cite{johannesson15b}, or cooperative energy management of multiple vehicles \cite{murgovski16,hovgard18}. In \cite{uebel17} a PMP-DP method has been proposed to solve the optimal control of vehicle speed, battery energy, gear selection and ICE on/off state. However, the computational effort of the control algorithms is still highly susceptible to long horizon lengths and high update frequencies.

The synergy among different optimisation methods is generally performed by splitting the problem into sub-problems arranged into multi-level or bi-level control architectures, where different tasks are delegated to distinct layers based on horizon length, time constants, sampling interval and updating frequency. To this end, multi-level and bi-level model predictive control (MPC) algorithms have been proposed for conventional vehicles (CVs), \cite{guo18}, and HEVs, \cite{turri16,guo16,stroe17,uebel19}, respectively. The multi-level architectures allow solving computationally intensive sub-problems, e.g. mixed-integer programs. Such programs are typically solved by an MPC, tracking a certain reference or a target state, typically over look-ahead horizons of up to \SI{20}{km}. Even though such horizons may appear long, there are problems that are naturally defined for even longer horizons. 

Problems with very long look-ahead horizons, in the order of hundreds of kilometres, are typically addressed in logistics \cite{hanson17}. As an example of a target state over long horizon is the travel time, which is often given at the end of the route. In the case of electrified vehicles, a target battery state of charge may also be provided at charging locations along the route. 
Within the multi-level control architecture mentioned earlier, these problems are delegated to the highest supervisory level, generating reference travel time and battery state of charge trajectories over hundreds of kilometres. Early results on developing online implementable controllers that operate over long horizons, hereafter referred to as the mission managers, have been published in our previous work for the case of CV, see \cite{hamednia18}. 

The goal of this paper is to generalise the mission manager developed in \cite{hamednia18} to both CVs and EVs. The purpose of the mission manager is to generate optimal reference trajectories for the entire route, or for look-ahead horizons that may stretch over hundreds of kilometres. 
The computational effort are decreased in three steps: 1) a problem decomposition into two sub-problems, where velocity and travel-time trajectory are optimised online and gear shifting strategy is optimised offline; 2) a combination of an indirect PMP solution and a direct nonlinear programming for reducing the number of states in the online optimisation sub problem; 3) a real-time iteration (RTI) sequential quadratic programming (SQP) \cite{diehl01}, which allows a single quadratic program (QP) to be solved in
an MPC manner \cite{thomas94}. 

The outline of the paper is as follows. In Section \ref{sec:model}, dynamic model of vehicle is presented. In Section \ref{sec:problem_formulation}, the energy minimisation problem is formulated. Section \ref{sec:cea} describes the computationally efficient algorithm. In Section \ref{sec:cases} the proposed algorithm is applied to a CV and an EV. In Section \ref{sec:res}, the simulation results are demonstrated. Finally, Section \ref{sec:conclusion} concludes the paper. 

\section{Physical Modelling}\label{sec:model}

This section addresses vehicle dynamics, i.e. travel time and longitudinal vehicle dynamics. Furthermore, static relations are given that translate torque and rotational speed of actuator to traction force and longitudinal velocity. Finally, lower bounds and upper bounds on longitudinal velocity, traction force and acceleration are presented.


\subsection{Travel time and longitudinal dynamics}\label{subsec:time_dyn}

According to Newton's law of motion, preliminary governing equations of a point mass vehicle model are \begin{align}
    &\dot s(t)=v(t) \label{eq:time_dyn_t}\\
    &\tx{m}\,\dot v(t)=F(t)+F_\tx{brk}(t)-F_\tx{air}(v)-F_\alpha(s)\label{eq:long_dyn_t}
\end{align} 
where $\tx{m}$ is total lumped mass of the vehicle, $t$ is travel time, $s$ is travelled distance, $v$ is longitudinal velocity, $F$ is non-negative traction force at the wheel side of the vehicle generated by the actuator, and $F_\tx{brk}$ is a non-positive force that includes braking by the service brakes, a retarder, a compression release engine brake and/or an exhaust pressure governor. For the case of a conventional vehicle, more details on the braking force will be discussed later, in Section~\ref{subsec:on-cv}. Note that the travelled distance and longitudinal velocity are functions of travel time in \eqref{eq:long_dyn_t}. However, the explicit dependence is not shown for brevity, when these signals are input arguments to functions, such as $F_\alpha(s(t))$ and $F_\tx{air}(v(t))$. The nominal aerodynamic drag, $F_\tx{air}$, and resistive forces that depend on road gradient $\alpha$, $F_\alpha(s)$, are defined as
\begin{align}
    &F_\tx{air}(v)=\frac{\uprho_\tx{a}\tx{c}_\tx{d}\tx{A}_\tx{f}v^2}{2}, \label{eq:aero}\\
    &F_\alpha(s)=\tx{mg}\left(\sin(\alpha(s))+\tx{c}_\tx{r}\cos(\alpha(s))\right), \label{eq:rol_pot}
\end{align}
where $\uprho_\tx{a}$ is air density, $\tx{c}_\tx{d}$ is aerodynamic drag coefficient, $\tx{A}_\tx{f}$ is vehicle frontal area, g is the gravitational acceleration, and $\tx{c}_\tx{r}$ is rolling resistance coefficient. 

The vehicle longitudinal dynamics \eqref{eq:time_dyn_t} and \eqref{eq:long_dyn_t}, are nonlinear due to the quadratic dependency of longitudinal velocity in the aerodynamical drag function in \eqref{eq:aero} and the road gradient that can be an arbitrary nonlinear function of distance in \eqref{eq:rol_pot}. The nonlinearity may increase computational complexity. To overcome this issue, it is possible to modify the equations in \eqref{eq:time_dyn_t} and \eqref{eq:long_dyn_t} by changing independent variable and change state variables. If distance $s$ is used as independent variable instead of time $t$ in \eqref{eq:time_dyn_t}, i.e. decisions are planned with respect to $s$. 
This issue can be overcome by changing the independent variable from time to distance, i.e. decisions are planned with respect to $s$ instead of $t$, as presented in \cite{murgovski13china,lipp14,murgovski15ACC,castro16}. Subsequently, for a given road topography, the function $F_\alpha$ now becomes a fixed trajectory for the entire route. In addition, the nonlinearity in \eqref{eq:aero} can be removed by a change of state variable $v$ to kinetic energy, 
\begin{align}
    E(s)=\frac{\tx{m}v^2(s)}{2} \label{eq:v2E}
\end{align}
where $E$ represents the kinetic energy of the vehicle. These transformations are non-approximate as long as the studied vehicle does not stop or change direction of its movement. 
Also, to study variations on acceleration and jerk of the driving vehicle, we introduce acceleration, $a$, as an additional state variable. The change of acceleration in space coordinates, which resembles jerk, $j$, now becomes the input signal to the vehicle system. The resulting vehicle dynamics model becomes
\begin{align}
    &t'(s)=\sqrt{\frac{\tx{m}}{2E(s)}}\label{eq:time_dyn_s}\\
    &E'(s)=\tx{m}a(s) \label{eq:a_dyn_s}\\
    &a'(s)=j(s) \label{eq:j_dyn_s}
\end{align}
where $t'$ and $a'$ are used as short hand notations for $\tx{d}t/\tx{d}s$ and $\tx{d}a/\tx{d}s$, respectively. The $E'=\nobreak \tx{m}vv'$ is the product of mass and vehicle acceleration, and
\begin{align}
    & a(s)=\frac{1}{\tx{m}}\left(-\tx{c}_\tx{a}E(s)+F(s)+F_\tx{brk}(s)-F_\alpha(s)\right)
    \label{eq:F2a}
\end{align}
where $\tx{c}_\tx{a}=\uprho_\tx{a}\tx{c}_\tx{d}\tx{A}_\tx{f}/2$ gathers the drag related coefficients.

It can be noticed that \eqref{eq:time_dyn_s} is still nonlinear with respect to $E$. 
More information on how to tackle the nonlinearity in \eqref{eq:time_dyn_s} is presented in Section \ref{sec:cea}.

Throughout this paper, all constants, which are not dependent on $s$ are shown in upright letters, e.g. $\tx{m},\tx{A}_\tx{f},\tx{c}_\tx{d},\uprho_\tx{a}$  do not depend on $s$. However, all the states and control inputs are trajectories in terms of $s$, e.g. $t(s)$ and $E(s)$ are trajectories dependent on $s$, where in several places the dependency is not displayed for simplicity.

\begin{figure}
 \centering
 \includegraphics[width=0.95\linewidth]{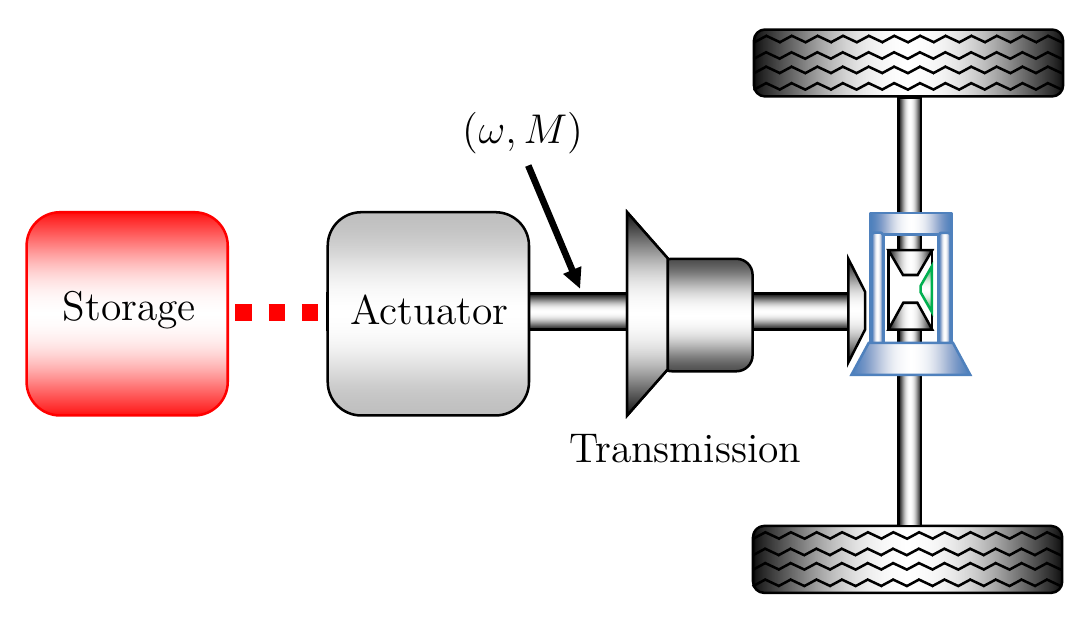}\vspace{-.3cm}
   \caption{\footnotesize Schematic diagram of the studied powertrain. The powertrain consists of energy storage unit, actuator and transmission system, which transfers shaft torque, $M$, with rotating speed $\omega$.}
   \label{fig:schematic_gen}
\end{figure}

\subsection{Vehicle powertrain}\label{subsec:powertrain}
A schematic diagram of the considered powertrain is illustrated in Fig. \ref{fig:schematic_gen}. The powertrain consists of an energy storage unit, an actuator, e.g. an ICE or an electric machine (EM), and a transmission system. The torque and speed at the shaft between the actuator and transmission is denoted by $M$ and $\omega$, respectively.

The transmission system is modelled considering the transmission and final gear ratios as 
\begin{align}
    v(s)=\omega(s) R(\gamma), \quad F(s)=\frac{M(s)}{R(\gamma)},
    \label{eq:trn}
\end{align}
where $\gamma$ denotes selected gear, and 
\begin{align}
    R(\gamma)=\frac{\tx{r}_\tx{w}}{\tx{r}_{\tx{tg}}(\gamma)\tx{r}_{\tx{fg}}}
    \label{eq:ratio}
\end{align}
where $\tx{r}_\tx{w}$ is the wheel radius, $\tx{r}_\tx{tg}$ and $\tx{r}_{\tx{fg}}$ are transmission and final gear ratios, respectively.



The speed limits 
\begin{align}
    &v_\tx{min}(s)=\max\left\{v^{\tx{road}}_{\tx{min}}(s), v^{\tx{traffic}}_{\tx{min}}(s)\right\},\\
    &v_\tx{max}(s)=\min\left\{v^{\tx{road}}_{\tx{max}}(s), v^{\tx{traffic}}_{\tx{max}}(s)\right\},
\end{align}
are obtained by considering legal speed limits of the road, $v^{\tx{road}}_{\tx{min}}(s)$ and $v^{\tx{road}}_{\tx{max}}(s)$ and dynamic traffic speed limits, $v^{\tx{traffic}}_{\tx{min}}(s)$ and $v^{\tx{traffic}}_{\tx{max}}(s)$.


The traction force limits as functions of kinetic energy are
\[ F(s) \in [F_{\gamma\tx{min}}(E), F_{\gamma\tx{max}}(E)], \] 
where 
\begin{align}
    & F_{\gamma\tx{min}}(E)=\min_\gamma F_\tx{min}(E,\gamma),\\
    & F_{\gamma\tx{max}}(E)=\max_\gamma F_\tx{max}(E,\gamma).
\end{align}
The functions $F_\tx{min}(E,\gamma)$ and $F_\tx{max}(E,\gamma)$ are the traction force limits for a given pair of kinetic energy (longitudinal velocity) and gear.

In turn, the acceleration limits, 
\[ a(s)\in[a_\tx{min}(E),a_\tx{max}(E)], \]
can be derived using \eqref{eq:F2a} as a function of kinetic energy (longitudinal velocity) and considering the limits on 
traction force, as
\begin{align}
    &a_\tx{min}(E)=\max\left\{\underline{\tx{a}}, \frac{F_{\gamma\tx{min}}(E) -\tx{c}_\tx{a}E+\underline{\tx{F}_\tx{brk}}-F_\alpha}{\tx{m}} \right \} \label{eq:a_min}\\
    &a_\tx{max}(E)=\min\left\{\overline{\tx{a}},\frac{F_{\gamma\tx{max}}(E) -\tx{c}_\tx{a}E-F_\alpha}{\tx{m}} \right\} \label{eq:a_max}
\end{align}
where $\underline{\tx{a}}$ is the minimum and $\overline{\tx{a}}$ is the maximum allowed acceleration within a comfort zone and $\underline{\tx{F}_\tx{brk}}$ denotes constant minimum total braking force. Here, $a_\tx{min}$ and $a_\tx{max}$ are not necessarily smooth functions, as
$F_{\gamma\tx{min}}$ and $F_{\gamma\tx{max}}$ may not be smooth functions. This will be discussed in more details in Section~\ref{sec:cases}.

In order to deliver a certain traction force, the actuator draws power from the energy storage unit. Let $P_\tx{w}(v,F,\gamma)$ denote the drawn power, which in the case of a combustion engine is a chemical, fossil fuel power, and in the case of an electric machine, it is an electric power. Explicit representations of the internal power in terms of the kinetic energy (longitudinal velocity) and traction force will be provided later, in Section \ref{sec:cases}.

\section{PROBLEM FORMULATION}\label{sec:problem_formulation}

This section formulates an optimisation problem, which aims at planning optimal velocity trajectory for the entire route, in a way that total energy consumption is minimised and the travel time is upper bounded. A performance function is formulated as
\begin{align}
    &\int_{0}^{\tx{s}_\tx{f}} \left(\frac{\tx{c}_\tx{eg}P_\tx{w}(v,F,\gamma)}{v(s)} + \tx{w}_1 a^2(s)+\tx{w}_2 j^2(s) \right)\tx{d}s
    \label{eq:p1_cost_v}
\end{align}
that minimises total energy consumption by integrating the internal power drawn from the storage unit and the driver's discomfort via the acceleration and jerk, where $\tx{c}_\tx{eg}$ is the price of energy storage, and $\tx{w}_1$ and $\tx{w}_2$ are penalty factors. The division of the internal power with speed in \eqref{eq:p1_cost_v} derives from the time to space transformation,
\[\int P_\tx{w}(v, F,\gamma) \tx{d}t=\int P_\tx{w}(v, F,\gamma)/v(s) \,\tx{d}s. \]

\subsection{Computing upper bound on travel time}\label{subsec:v_est}
To compute the upper bound on travel time, $\tx{t}_\tx{f}$, it is possible to obtain a velocity profile, $v_{\tx{hg}}(s)\in[v_{\tx{min}}(s), v_{\tx{max}}(s)]$, as a heuristic guess by filtering cruising speed, $\tx{v}_{\tx{cru}}\in[v^{\tx{road}}_{\tx{min}}(s), v^{\tx{road}}_{\tx{max}}(s)]$. The rated power of the actuator and road/traffic limits on speed are taken into consideration in the cruise speed filtering \cite{murgovski16,johannesson15a,hamednia18}. The cruising speed is assumed to be set manually by the driver or automatically by a telemetry system. 
When deriving $v_{\tx{hg}}$, it is assumed that the vehicle will try to maintain $\tx{v}_{\tx{cru}}$ unless actuator limit is reached.
\begin{equation}
\begin{split} 
    v_{\tx{hg}}(s)&=\min\Bigg\{\tx{v}_{\tx{cru}}, \int_{0}^s
    \frac{a_\tx{max}(v_{\tx{hg}}) }{v_{\tx{hg}}(\sigma)} \tx{d}\sigma \Bigg\}
\end{split} \label{eq:Vest}
\end{equation}
By computing maximum arrival time as
\begin{align} \label{eq:final_travel_time}
    \tx{t}_\tx{f}=\int_{0}^{\tx{s}_\tx{f}} \frac{\tx{d}s}{v_{\tx{hg}}(s)},
\end{align}
where $\tx{s}_\tx{f}$ is the final position at the end of the route, a constraint can be imposed 
\begin{align}
    t(\tx{s}_\tx{f}) \leq \tx{t}_\tx{f} \label{eq:final_time}
\end{align}
that requires finishing the route in the same time or sooner than what would be required when driving with $v_{\tx{hg}}$.

\subsection{MPC for minimizing energy consumption}

The problem \eqref{eq:p1_cost_v} is optimised in an MPC framework with a prediction horizon of length $s_\tx{H}$. The goal of this paper is to develop a computationally efficient algorithm that allows horizons that cover the entire route. However, as computational resources are always limited, we impose an upper bound, $\tx{s}_\tx{Hmax}$, hopefully in the range of hundreds of kilometres. The optimisation problem can then be solved in a moving horizon MPC (MHMPC) framework if ${\tx{s}_\tx{Hmax} < \tx{s}_\tx{f}}$, or in a shrinking horizon MPC (SHMPC) framework if ${\tx{s}_\tx{Hmax} \geq \tx{s}_\tx{f}}$.
The optimisation variables are predicted at distance samples ${s\in[\zeta, \zeta + s_\tx{H}]}$, given information of the actual vehicle's states at $\zeta$. Thus, the actual horizon length can be computed as
\begin{align}
    s_\tx{H}(\zeta)=\min\{\tx{s}_\tx{Hmax},\tx{s}_\tx{f}-\zeta\}.
\end{align}

The problem can now be summarised as follows
{\allowdisplaybreaks
\begin{subequations} \label{eq:p1}
\begin{align}
\begin{split}
&\min_{j,F_\tx{brk},\gamma} \int_{\zeta}^{\zeta+s_\tx{H}(\zeta)} \Bigg(\frac{\tx{c}_\tx{eg}P_\tx{w}(E,F,\gamma)}{\sqrt{\frac{2E(s|\zeta)}{\tx{m}}}}\\
&\hspace{30mm}+ \tx{w}_1 a^2(s|\zeta)+\tx{w}_2 j^2(s|\zeta) \Bigg)\tx{d}s,
\end{split} \label{eq:p1_cost_E}\\
& \text{subject to:} \nonumber\\
&t'(s|\zeta)=\sqrt{\frac{\tx{m}}{2E(s|\zeta)}} \label{eq:p1_time}\\
&E'(s|\zeta)=\tx{m}a(s|\zeta) \label{eq:p1_E}\\
&a'(s|\zeta)=j(s|\zeta) \label{eq:p1_a}\\
&F(s|\zeta)=\tx{m}a(s|\zeta)+\tx{c}_\tx{a}E(s|\zeta)-F_\tx{brk}(s|\zeta)+F_\alpha(s)\\
& E(s|\zeta)\in \frac{\tx{m}}{2} [v_\tx{min}^2(s|\zeta), v_\tx{max}^2(s|\zeta)] \label{eq:p1_Ebound}\\
& a(s|\zeta)\in [a_\tx{min}(E),a_\tx{max}(E)]\label{eq:p1_abound} \\
& j(s|\zeta)\in [\underline{\tx{j}}, \overline{\tx{j}}]\label{eq:p1_jbound}\\
& F_\tx{brk}(s|\zeta) \in [\underline{\tx{F}_\tx{brk}},0]\label{eq:p1_Fbrkbound} \\
& t(\zeta|\zeta)=t_\tx{0}(\zeta), \quad E(\zeta|\zeta)=E_\tx{0}(\zeta), \quad a(\zeta|\zeta)=a_\tx{0}(\zeta) \label{eq:p1_int}\\
& t(\zeta+s_\tx{H}|\zeta) \leq t_\tx{H}(\zeta) \label{eq:p1_tf}\\
& \gamma(s|\zeta) \in \{1, 2,\dots, \upgamma_\tx{max} \} \label{eq:p1_gear_cstr}
\end{align}
\end{subequations}}%
where $\underline{\tx{j}}$ is the minimum and $\overline{\tx{j}}$ is the maximum allowed jerk within a comfort zone, $t_\tx{0}$, $E_\tx{0}$ and $a_\tx{0}$ are the values of the system states at instant $\zeta$, and $\upgamma_{\tx{max}}$ is the highest gear. 
The constraints \eqref{eq:p1_time}-\eqref{eq:p1_gear_cstr} are enforced for all ${s\in[\zeta,\zeta+s_\tx{H}(\zeta)]}$ and the problem is re-evaluated for all ${\zeta \in [0, \tx{s}_\tx{f}]}$. The maximum allowed travel time over the prediction horizon, $t_\tx{H}$, is computed as in \eqref{eq:final_travel_time} for the distance $s_\tx{H}$. The problem \eqref{eq:p1} is a non-convex, mixed-integer and dynamic nonlinear program, where $t$, $E$ and $a$ are real-valued state variables, $j$ and $F_{\tx{brk}}$ are real-valued control inputs, $\gamma$ is an integer control input and $F$ is an output variable. Although from a  control point of view $j$ is the control signal, in practice, the acceleration $a$ is applied to the vehicle.

For the sake of simplicity, the dependence on $\zeta$ will not be shown in most following parts of the paper and the method is explained via a single MPC update, e.g. the one with $\zeta=0$.

\section{Computationally Efficient Algorithm}\label{sec:cea}
This section proposes reformulation  steps of the problem \eqref{eq:p1} to enhance the computational efficiency. These steps are: 1) bi-level optimisation program that allows decoupling the integer variable, i.e. gear, from a nonlinear optimisation program (NLP); 2) adjoining nonlinear dynamics of travel time to the objective using necessary PMP conditions for optimality; 3) Removing a loop on finding optimal time costate and applying RTI SQP scheme.
\subsection{Bi-level programming and gear optimisation}\label{sec:bi-level}

The mixed-integer problem \eqref{eq:p1} can be reformulated as a bi-level program: 
{\allowdisplaybreaks
\begin{subequations} \label{eq:p2}
\begin{align}
&\min_{j,F_\tx{brk}} \int_{0}^{s_\tx{H}} \Bigg(\frac{\tx{c}_\tx{eg}P_\tx{w}(E, F,\gamma^*)}{\sqrt{\frac{2E(s)}{\tx{m}}}} + \tx{w}_1 a^2(s) + \tx{w}_2 j^2(s)\Bigg)\tx{d}s \label{eq:p2_cost}\\
&  \text{subject to:\eqref{eq:p1_time}-\eqref{eq:p1_tf}} \nonumber\\
& \gamma^*(s) = \argmin_{\gamma} P_\tx{w}(E,F,\gamma)\\
& \text{subject to: } \gamma(s) \in \{1, 2, ..., \upgamma_\tx{max} \}\\
& \hspace{18mm} F(s)\in [F_{\gamma\tx{min}}(E),F_{\gamma\tx{max}}(E)]\label{eq:p2_Fbound}
\end{align}
\end{subequations}}%
where optimisation of gear resides only in the bottom level program, while all the system dynamics reside in the top level program. If the actuator and transmission system are modelled statically, it is possible to separate the bottom level and solve offline, if $E$ and $F$ are regarded as parameters and optimal gear is computed as a function of the parameters. To this end, the bottom level can be solved as
{\allowdisplaybreaks
\begin{subequations} \label{eq:opt_gear}
\begin{align}
& f^*_{\gamma}(E,F) = \argmin_\gamma P_\tx{w}(E,F,\gamma)\\
& \text{subject to: } \gamma \in \{1, 2, ..., \upgamma_\tx{max} \}\\
& \hspace{17mm} F \in \mathcal{F}(E)=[F_{\gamma\tx{min}}(E), F_{\gamma\tx{max}}(E) ]\\
& \hspace{17mm} E \in \mathcal{E}(\gamma)=\frac{\tx{m}[\upomega_\tx{idle}^2, \upomega_\tx{max}^2 ] \tx{R}^2(\gamma)}{2}
\end{align}
\end{subequations}}%
where $f^*_{\gamma}(E,F)$ is a two-dimensional function describing the optimal gear choices for all traction force versus speed (kinetic energy) combinations, $\mathcal{E}$ and $\mathcal{F}$ are the feasible sets for kinetic
energy and traction force respectively, and $\upomega_\tx{idle}$ and $\upomega_\tx{max}$ are rotational speed limits. 
By replacing the optimal gear with the parametric function, the internal power can be written as 
\begin{align}
    \label{eq:opt_p}
    &P_{\gamma}(E,F) = P_\tx{w}(E,F,f^*_{\gamma}(E,F)),
\end{align}
indicating power consumption when gear is optimally chosen. Note that for CV case study the offline-optimised gear selection algorithm is extended, which covers the negative force area originating from negative additional force. More details will be given later in Section. \ref{sec:cases}.

\subsection{Necessary PMP conditions for optimality} \label{subsec:non-app-pmp}

In the second step of the algorithm, the problem \eqref{eq:p2} is reformulated, which is facilitated by the necessary PMP conditions for optimality. The Hamiltonian is defined as
\begin{align}
\begin{split}
    \mathcal{H}(\cdot) &= \tx{c}_\tx{eg}P_{\gamma}(E,F) \sqrt{ \frac{\tx{m}}{2E(s)} } + \tx{w}_1 a^2(s)+\tx{w}_2 j^2(s) +\\ &\hspace{4mm}+\lambda_\tx{t}(s) \sqrt{\frac{\tx{m}}{2E(s)}} + \lambda_E(s) \tx{m}a(s)+\lambda_a(s) j(s).
\end{split}
\end{align}
where the symbol $\cdot$ is a  compact notation for a function of multiple variables. Here, $\lambda_\tx{t}$, ${\lambda}_E$ and $\lambda_a$ denote the costates of travel time, kinetic energy and acceleration, respectively. It can been observed that the Hamiltonian is not an explicit function of travel time, thus the optimal time costate, $\lambda^*_t$, i.e. the value for $\lambda_\tx{t}$ that satisfies the maximum travel time constraint \eqref{eq:p1_tf}, is a constant value. Hence
\begin{align}
\label{eq:nec_cond}
    & {\lambda'}_t^*(s) = -\left( \frac{\partial \mathcal{H}(\cdot)}{\partial t}\right)^*=0.
\end{align}
Furthermore, the travel time is a strictly monotonically increasing function that may activate constraint \eqref{eq:p1_tf} only at the final instant. Consequently, if $\lambda_\tx{t}^*$ is known, it will be possible to remove the nonlinear constraint on travel time \eqref{eq:p2} and adjoin the product of $\lambda^*_t(s)$ and the nonlinear function $\sqrt{\frac{\tx{m}}{2E(s)}}$ to the objective function. This implies that the dynamic optimal control problem can yet again be formulated as a bi-level program
{\allowdisplaybreaks
\begin{subequations} \label{eq:p3}
\begin{align}
\begin{split} 
&\min_{\lambda_\tx{t}} \int_{0}^{s_\tx{H}} \Bigg( \frac{\tx{c}_\tx{eg}P_{\gamma}(E^*(\lambda_\tx{t},s),F^*(\lambda_\tx{t},s))+\lambda_\tx{t}}{\sqrt{\frac{2E^*(\lambda_\tx{t},s)}{\tx{m}}}}\\
&\hspace{25mm}+\tx{w}_1 {a^*}^2(\lambda_\tx{t},s)+\tx{w}_2 {j^*}^2(\lambda_\tx{t},s) \Bigg)\tx{d}s
\end{split} \label{eq:p3_top_cost}\\
&\text{subject to:} \nonumber\\
&{t'}^*(\lambda_\tx{t},s)=\sqrt{\frac{\tx{m}}{2E^*(\lambda_\tx{t},s)}}\label{eq:p3_top_t}\\
&{E'}^*(\lambda_\tx{t},s)=\tx{m} a^*(\lambda_\tx{t},s)\label{eq:p3_top_E}\\
&{a'}^*(\lambda_\tx{t},s)=j^*(\lambda_\tx{t},s)\label{eq:p3_top_a}\\
& t^*(\lambda_\tx{t},0)=t_\tx{0}, \quad t^*(\lambda_\tx{t},s_\tx{H}) \leq t_\tx{H}\label{eq:p3_top_t0H}\\
\begin{split}
&[j^*(\lambda_\tx{t},s), F_\tx{brk}^*(\lambda_\tx{t},s),F^*(\lambda_\tx{t},s)]=\argmin_{j,F_\tx{brk}}
\\ 
&\hspace{5mm}\int_{0}^{s_\tx{H}} \Bigg(\frac{\tx{c}_\tx{eg}P_{\gamma}(E,F)+\lambda_\tx{t}}{\sqrt{\frac{2E(s)}{\tx{m}}}}+\tx{w}_1 a^2(s)+\tx{w}_2 j^2(s) \Bigg)\tx{d}s 
\end{split} \label{eq:p3_bottom}\\
&\text{subject to: \eqref{eq:p1_E}-\eqref{eq:p1_Fbrkbound},$\quad E(0)=E_0, \quad a(0)=a_0$} \nonumber
\end{align}%
\end{subequations}}%
where all constraints involving travel time have been moved to the top level, while the bottom level, \eqref{eq:p3_bottom}, generates optimal control trajectories parameterised in $\lambda_\tx{t}$. 
Similarly as before, the goal is to separate the two optimisation levels. One way to do this is by trying different values for $\lambda_\tx{t}$ and then using search methods, e.g. Newton or bisection, to find $\lambda_\tx{t}^*$ that minimises the top level's cost.

By assuming that problem \eqref{eq:p3_bottom} is an NLP that can be solved with sequential quadratic programming (SQP), the procedure for solving the mixed-integer problem \eqref{eq:p1} will consist of three nested loops as illustrated in Fig.~\ref{fig:flwchrt}a. The outermost loop updates the MPC horizon, the middle loop finds the optimal value for $\lambda_\tx{t}$ and the innermost loop sequentially solves a QP in order to find the solution of problem \eqref{eq:p3_bottom} for a given value of $\lambda_\tx{t}$. The procedure is still computationally inefficient, as it requires solving multiple QPs for given multiple $\lambda_\tx{t}$ values in each MPC update. Our goal is to eliminate the inner most loops and for a given $\lambda_\tx{t}$, solve only a single QP in each MPC update, as illustrated in Fig~\ref{fig:flwchrt}b.
\begin{figure}
    \centering
    \includegraphics[width=\linewidth]{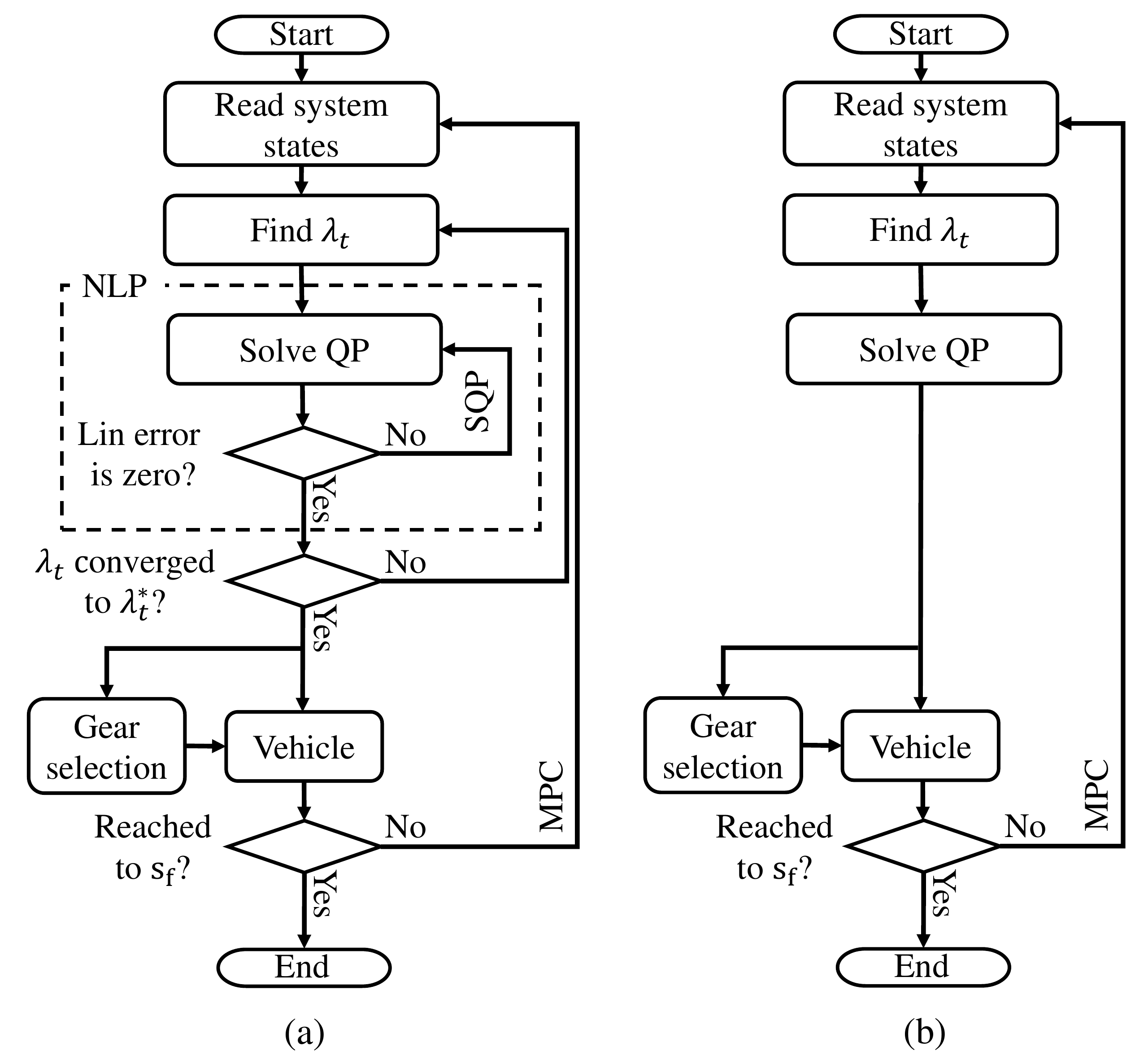}\vspace{-0.4cm}
    \caption{Flowchart of the proposed algorithm to solve NLP in MPC framework, (a) with three nested loops, innermost loop to solve NLP using SQP and middle loop to find $\lambda^*$; (b) using RTI SQP, which solves a single QP in each MPC update.}
    \label{fig:flwchrt}
\end{figure}

\subsection{Updating the time costate over the MPC loop}
To eliminate the loop on finding $\lambda^*_t$, it is considered that the optimal energy consumption corresponds in general to driving slow, so it can be assumed that the vehicle will use the entire travel time, i.e. ${t^*(\lambda_\tx{t},s_\tx{H})\approx t_\tx{H}}$. Hence, the objective of the top level program in \eqref{eq:p3} is transformed to minimising maximum travel time difference, as
\begin{align}
    \min_{\lambda_\tx{t}} ||t^*(\lambda_\tx{t},s_\tx{H}|\zeta) - t_\tx{H}(\zeta)|| 
    \label{eq:J_lambda}
\end{align}
where $||\cdot||$ may indicate any norm.

For the case that the problem \eqref{eq:p3} is solved in SHMPC framework, the final time instant and the final point of the horizon are fixed regardless of the update instant $\zeta$, i.e. $t_\tx{H}(\zeta)=\tx{t}_\tx{f}$ and $\zeta+s_\tx{H}(\zeta)=\tx{s}_\tx{f}$, $\forall \zeta$. 
 
\begin{lemma} \label{lemma_tstar}
If predicted disturbances do not change and there is no miss-match between the control and plant model, then for an SHMPC implementation of problem~\eqref{eq:p3} and for a given $\lambda_\tx{t}$, it holds,
\begin{align}
    t^*(\lambda_\tx{t},s_\tx{H}|\zeta)=t^*(\lambda_\tx{t},s_\tx{H}|\zeta+\delta\zeta), \quad \forall \delta\zeta \in [0,\tx{s}_\tx{f}-\zeta],
\end{align}
i.e. the optimal travel time at the end of the horizon does not change for different SHMPC updates. 
\label{lem:1}
\end{lemma}

\begin{proof}
The proof follows directly from Bellman's principle of optimality, i.e. \textit{any tail of an optimal trajectory is an optimal solution as well} \cite{bellman57}.
\end{proof}

For an MHMPC, Lemma~\ref{lemma_tstar} does not hold even if disturbances are predicted exactly and there is no model miss-match. This is because new information is added as the prediction horizon moves forward at each MPC update. 
However, if the prediction horizon is much longer than the interval between two consecutive updates, then for different $\zeta$, it can be assumed
\begin{align}
    t^*(\lambda_\tx{t},s_\tx{H}|\zeta)-t_\tx{H}(\zeta)\approx t^*(\lambda_\tx{t},s_\tx{H}|\zeta^+)-t_\tx{H}(\zeta^+)
\end{align}
where $\zeta^+$ is the instance of the MHMPC update following that at $\zeta$. 
Fig. \ref{fig:res3B_cvev} demonstrates the overlapped curves of the final time difference versus the time costate for a CV and an EV, where $\zeta=\SI{0}{m}$ and $\zeta^+=\SI{300}{m}$. Thus, it is also possible for an MHMPC to update the time costate over the MPC loop.

\begin{figure}[t]
\centering
\begin{tikzpicture}
\node[rotate=0] at(0,0){\includegraphics[width=0.95\linewidth]{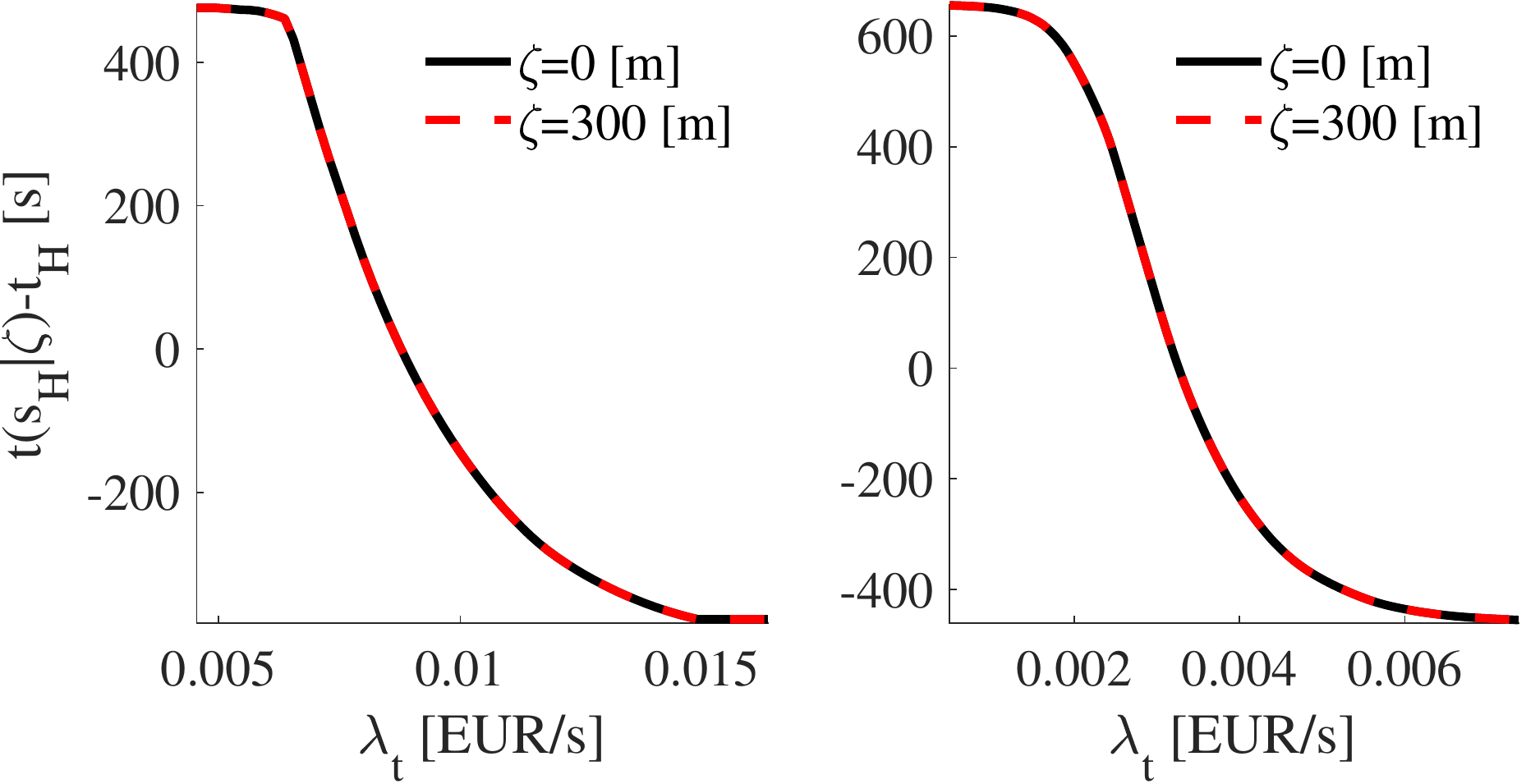}};
\node at(-1.4,-2.3){\footnotesize (a)};
\node at(2.85,-2.3){\footnotesize (b)};
\end{tikzpicture}
\vspace{-.4cm}
\caption{\footnotesize Difference between calculated time at the end of horizon and the desired maximum time for varying time costate using MHMPC scheme, where $\zeta=\SI{0}{m}$ and $\zeta^+=\SI{300}{m}$. Overlap of the curves for different $\zeta$ values shows that $\lambda_\tx{t}$ can be evaluated only once per each MPC update, rather than waiting for a full convergence. The curves are derived for, (a) a CV, and (b) an EV.}
\label{fig:res3B_cvev}
\end{figure}

Problem~\eqref{eq:J_lambda} is then solved by a derivative free Newton method, where the Newton iterates are spread across the MPC updates without waiting for a full convergence, i.e. by performing one Netwon step per update. A flowchart of the proposed algorithm is depicted in Fig.~\ref{fig:flwchrt}b, while more details on the Newton method is provided in Appendix~\ref{ap:lambda_opt}. 

\subsection{Real-time iterations SQP over the MPC loop}\label{subsec:rti}
For a given $\lambda_\tx{t}$ it remains to solve problem~\eqref{eq:p3_bottom}. It will be shown later, in Section~\ref{sec:cases}, that for the case of conventional and electric vehicle powertrians, problem~\eqref{eq:p3_bottom} is indeed a smooth NLP that can be solved by SQP. However, instead of sequentially solving a QP until linearization error is equal to zero, it is computationally efficient to spread the SQP over MPC updates, which is provided by RTI. The idea is to solve only a single QP per MPC update, without waiting for a full convergence. The obtained solution is possibly sub-optimal, but due to the contractivity of the RTI scheme as shown in \cite{diehl05}, the real-time iterates quickly approach the optimal solution during the runtime of the process. 

As the SQP is stopped prematurely, it is important to show that the obtained solution by solving a single QP is feasible in the original NLP. 
Feasibility can be guaranteed if the domain of the QP, obtained by linearizing the nonlinear constraints in problem~\eqref{eq:p3_bottom}, is an inner approximation of the feasible set of the NLP~\eqref{eq:p3_bottom}. This is indeed the case for conventional and electric vehicle powertrians, which will be shown in Section~\ref{sec:cases}.


\section{Application to CV and EV}\label{sec:cases}
This section proposes several steps that show how the computationally efficient algorithm proposed in Section~\ref{sec:cea} is applied to a CV and an EV.



\subsection{Conventional vehicle}\label{subsec:on-cv}

A conventional powertrain includes an ICE to transform chemical fuel energy to mechanical propulsion energy through a multiple-gear transmission.

A static fuel mass rate map for a given pair of rotational speed and engine torque is obtained by gathering steady-state data from a dynamic simulation model of a diesel engine, presented in \cite{wahlstrom11}. Subsequently, efficiency map and torque limits are derived, see Fig. \ref{fig:eta_cv}. According to the efficiency isolines, it is desirable to avoid operating the ICE at low speed and torque, where efficiency is low. 

Fig.~\ref{fig:eta_cv} also illustrates a negative torque limit for an additional braking system, including a retarder, a compression release engine brake and/or an exhaust pressure governor. The additional braking is preferred over the service braking in order to reduce wear and avoid lock up of the braking pads. 
Using \eqref{eq:trn}, the negative torque is translated to negative force on the wheel side as
\begin{align}
    F_\tx{brk} = F_\tx{A} + F_\tx{S},
\end{align}
where $F_\tx{S}$ and $F_\tx{A}$ are forces by the service brakes and the additional braking system. The minimum negative additional force limit for a given kinetic energy is
\begin{align}
    F_{\tx{Amin}}(E)= \min_\gamma F_{\gamma\tx{A}}(E,\gamma)
\end{align}
where $F_{\gamma\tx{A}}$ denotes the minimum negative additional force for each gear. The lower bound on the traction force is zero, i.e. $F_{\gamma\tx{min}}(E)=0$.

\begin{figure}
\centering
\begin{tikzpicture}
\node[rotate=0] at(0,0){\includegraphics[width=0.95\linewidth]{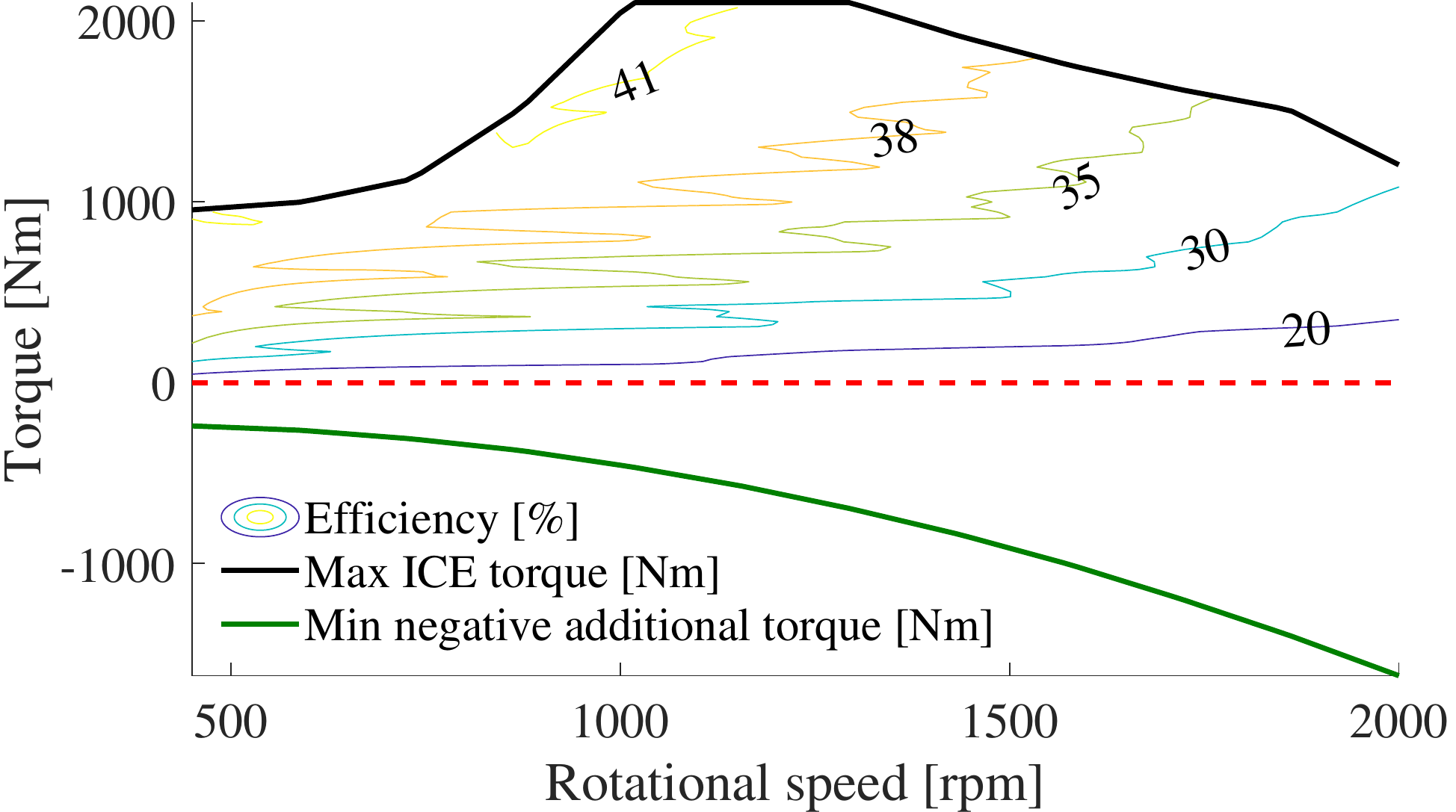}};
\end{tikzpicture}
\vspace{-.35cm}
\caption{\footnotesize Steady-state efficiency map and maximum torque limit of the ICE. The negative torque limit illustrates the braking capability of the additional braking system that includes a retarder, an exhaust pressure governor and/or a compression release engine brake.}
\label{fig:eta_cv}
\end{figure}

The two-dimensional fuel mass rate map of the ICE translates to a three-dimensional map on the wheels side. This three-dimensional map, denoted as  $\mu_\tx{w}(E,F,\gamma)$, can be expressed in terms of kinetic energy, traction force and gear using \eqref{eq:v2E} and \eqref{eq:trn}. Subsequently, a map, which represents the parametric internal power function, $P_\tx{w}(E,F,\gamma)$, can be derived as
\begin{align}
    P_\tx{w}(E, F,\gamma)=\mu_\tx{w}(E, F,\gamma)Q_\tx{lhv}
    \label{eq:p_mu}
\end{align}
where $Q_\tx{lhv}$ is diesel heating value.

The bi-level program \eqref{eq:p2}, can be extended for a CV case study, including the negative force region, which originates from the summation of negative additional force and service braking force, as
{\allowdisplaybreaks
\begin{subequations} \label{eq:p2cv}
\begin{align}
&\min_{j,F_\tx{brk}} \int_{0}^{s_\tx{H}} \Bigg(\frac{\tx{c}_\tx{eg}P_\tx{w}(E, F,\gamma^*)}{\sqrt{\frac{2E(s)}{\tx{m}}}} + \tx{w}_1 a^2(s) + \tx{w}_2 j^2(s)\Bigg)\tx{d}s \label{eq:p2cv_cost}\\
& \text{subject to: \eqref{eq:p1_time}-\eqref{eq:p1_tf}} \nonumber\\
& \resizebox{\linewidth}{!}{
$\gamma^*(s)=\begin{cases}
    \argmin_{\gamma} P_\tx{w}(E,F,\gamma), & \text{if $F+F_\tx{brk}\geq0$}.\\
    \argmax_{\gamma} F_{\gamma\tx{A}}(E,\gamma), & \text{if $F_\tx{Amin}$}(E)\leq F+F_\tx{brk}<0\\ \argmin_{\gamma} F_{\gamma\tx{A}}(E,\gamma), & \text{if $\underline{\tx{F}_\tx{brk}}\leq F+F_\tx{brk}< F_\tx{Amin}(E)$}
  \end{cases}$
  }\\
& \text{subject to: } \gamma(s) \in \{1, 2, ..., \upgamma_\tx{max} \}\\
& \hspace{18mm} F(s)+F_\tx{brk}(s)\in [\underline{\tx{F}_\tx{brk}},F_{\gamma\tx{max}}(E)]\label{eq:p2cv_Fbound}
\end{align}
\end{subequations}}%
Note that the traction force, $F$, and the total braking force, $F_\tx{brk}$, cannot have non-zero values simultaneously, i.e. it is not the case that $F>0$ and $F_\tx{brk}<0$ at the same time.

To approach the offline-optimal gear selection problem \eqref{eq:p2cv}, it is possible to grid the feasible sets of kinetic energy and total force, i.e. $F+F_\tx{brk}$. To this end, in the positive force region, for any feasible combination of longitudinal velocity (kinetic energy) and traction force, the optimal gear is the one that minimises energy consumption. 
In the negative force region, if the total demanded force is higher than the minimum negative additional force, the highest possible gear is selected, which avoids unnecessary down-shifting. However, if  total demanded force is lower than the minimum negative additional force, the lowest possible gear is selected, since it provides the most possible negative additional force, see Fig. \ref{fig:gear_cv}. The remaining demanded negative force is covered by the service brakes.

\begin{figure}
 \centering
 \includegraphics[width=0.95\linewidth]{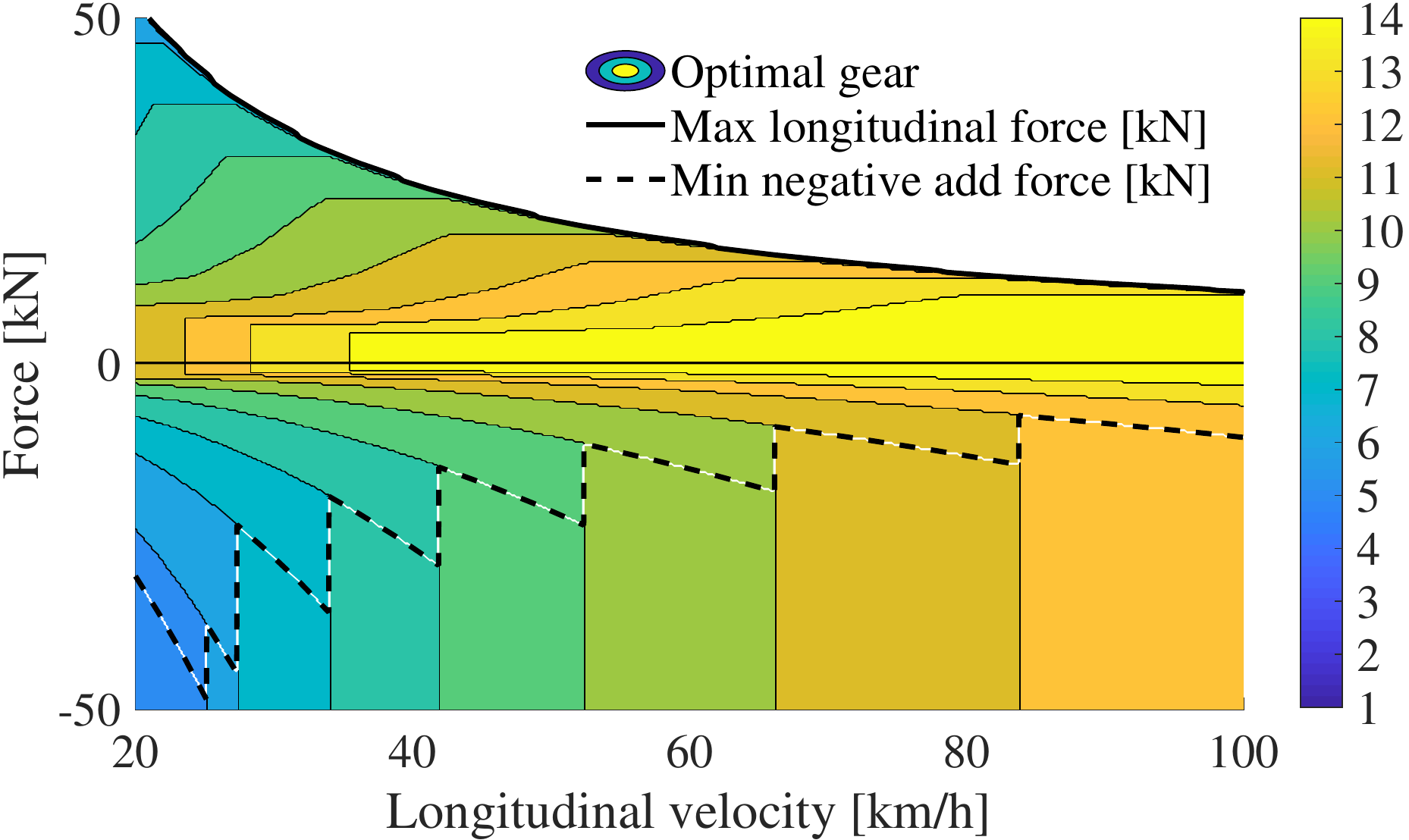}\vspace{-.3cm}
   \caption{\footnotesize Offline-optimised gear map together with maximum traction force and minimum negative additional force. In the positive force region, the optimal selected gear is the one that minimises fuel consumption, which for the studied powertrain coincides with the highest feasible gear. In the negative force region, if the total force is lower than the minimum negative additional force, the lowest possible gear is selected, since it provides the most possible negative additional force. The remaining demanded negative force is covered by the service brakes. However, if the total force is higher than the minimum negative additional force, to avoid unnecessary down-shifting, the highest possible gear is selected. 
   }
   \label{fig:gear_cv}
\end{figure}

\begin{figure}[t]
\centering
\begin{tikzpicture}
\node[rotate=0] at(0,0){\includegraphics[width=0.95\linewidth]{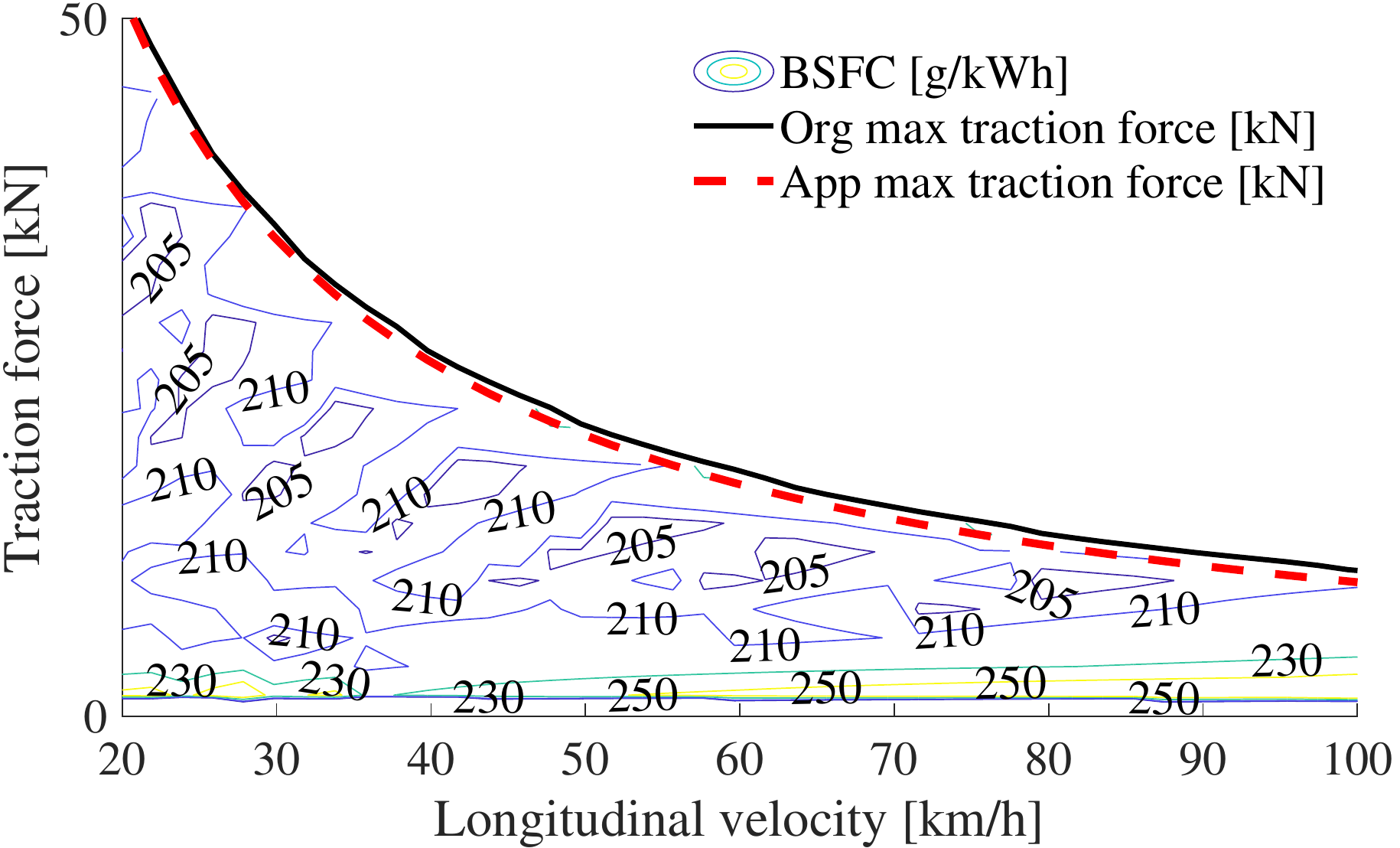}};
\end{tikzpicture}
\vspace{-.35cm}
\caption{\footnotesize Offline-optimised BSFC map together with original and approximate maximum traction force as well as minimum negative additional force. The approximate limit is an inner approximation for the longitudinal velocities above \SI{8}{km/h}.} 
\label{fig:opt_bsfc}
\end{figure}


The optimal brake specific fuel consumption (BSFC) map and maximum traction force curve are depicted in Fig. \ref{fig:opt_bsfc}. The optimal BSFC refers to the minimum burnt fuel, which is obtained by optimising the internal power in \eqref{eq:opt_p}.




The internal power drawn from fuel using \eqref{eq:opt_p}, is approximated by the following expression
\begin{align}
     P_{\gamma}(v,F) \approx \tx{p}_\tx{e0} + \tx{p}_\tx{e1} v^3(s) + \tx{p}_\tx{e2} v(s)F(s)\label{eq:exp_fv_cv}
\end{align}
with $\tx{p}_\tx{e0}, \tx{p}_\tx{e1}, \tx{p}_\tx{e2} \geq 0$. 
\begin{figure}
 \centering
 \includegraphics[width=0.95\linewidth]{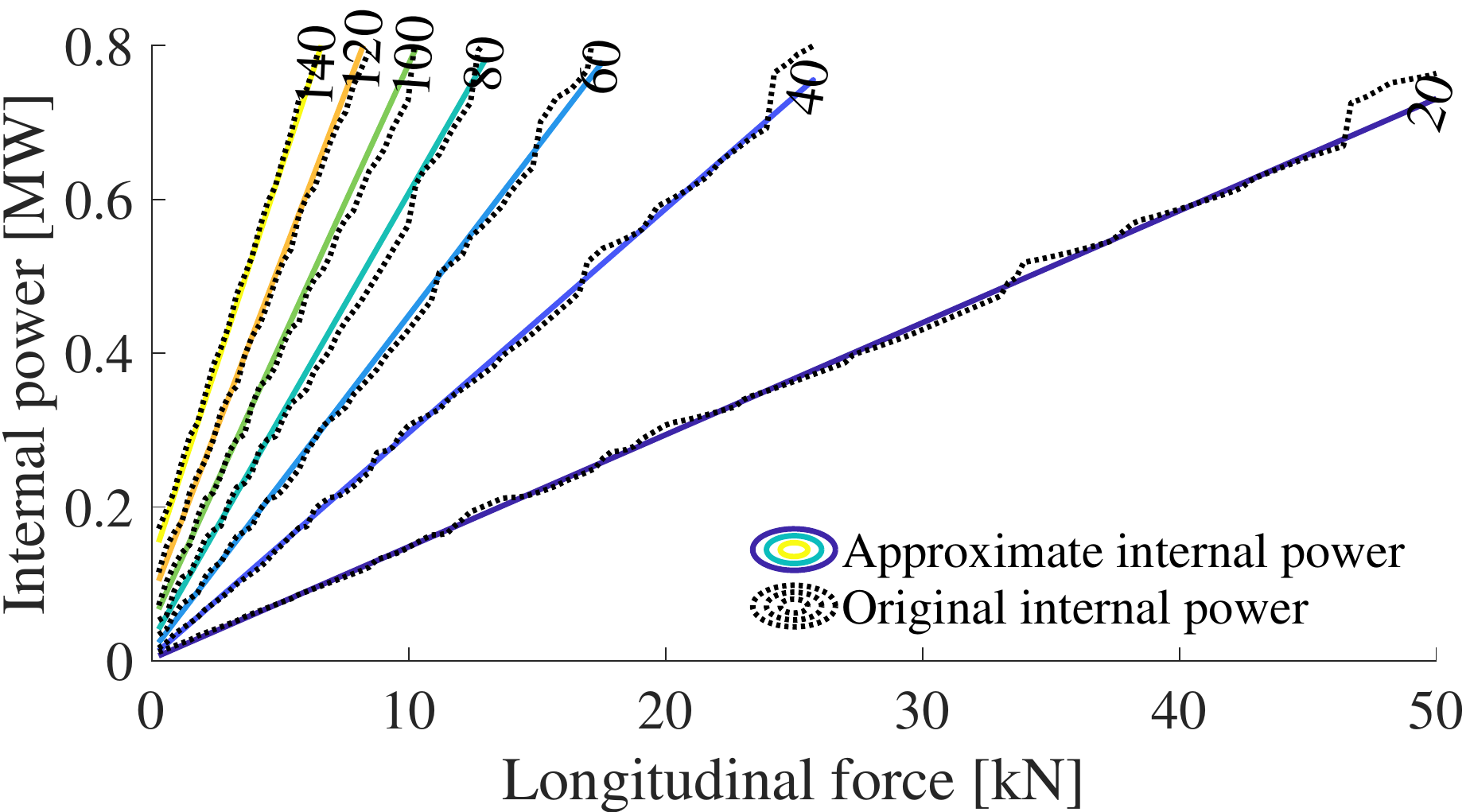}\vspace{-.3cm}
   \caption{\footnotesize Original and approximated internal power drawn from fuel for a given longitudinal velocity and traction force.}
   \label{fig:app_cv}
\end{figure}
As illustrated in Fig. \ref{fig:app_cv}, for the studied engine model it is sufficient to use a first order term in $F$, although it is possible to include higher order terms as well, without significant increase in computational effort. Similar expressions for model abstraction of fuel mass rate are exploited in \cite{murgovski16} and several references therein. Using \eqref{eq:v2E} and \eqref{eq:exp_fv_cv}, the stage cost \eqref{eq:p3_bottom} transforms into
\begin{equation}
\begin{split}
       V_{\tx{CV}}(\cdot,\lambda_\tx{t}) & \approx \frac{\tx{c}_\tx{eg}(\tx{p}_\tx{e0}+\lambda_\tx{t}^*)\sqrt{\tx{m}}}{\sqrt{2 E(s)}} + \frac{2\tx{p}_\tx{e1}}{\tx{m}} E(s) +\\ &+\tx{p}_\tx{e2}F(s)+\tx{w}_1 a^2(s)+\tx{w}_2 j^2(s)
       \label{eq:cv_cost_stage1}
\end{split}
\end{equation}
which is a convex second order cone function in terms of $E$, $a$, $j$, $F$ and $F_{\tx{brk}}$. 

The maximum traction force limit, see Fig. \ref{fig:opt_bsfc}, is approximated by
\begin{align}
    & F_{\gamma\tx{max}}(E)\approx \min\Bigg\{\overline{F},\tx{y}_\tx{0} + \frac{\tx{y}_\tx{1} \sqrt{\tx{m}}}{\sqrt{2E(s)}}\Bigg\}
    \label{eq:cv_force_max}
\end{align}
where $\overline{F}$ is the maximum constant traction force, and $\tx{y}_\tx{1}$ resembles the maximum engine power, as it can be alternatively written as a division of power with vehicle speed. The coefficients $\tx{y}_\tx{0}$ and $\tx{y}_\tx{1}$ are obtained by solving a linear program, see Appendix \ref{ap:lin_prg} for details. The approximated force limit \eqref{eq:cv_force_max} is an inner approximation of the original force for speeds above \SI{8}{km/h}, see Fig. \ref{fig:opt_bsfc}, which is acceptable for the highway scenarios investigated in this paper. 

The problem \eqref{eq:p3} with the stage cost \eqref{eq:cv_cost_stage1} is non-convex nonlinear program, because of the nonlinear term $\tx{y}_1/\sqrt{E(s)}$ in \eqref{eq:cv_force_max}. Due to the sign of $\tx{y}_1\geq 0$, this term is a convex function (a convex problem, though, requires a concave function here). It is possible to transform \eqref{eq:p3} to a convex second order cone program (SOCP) by linearizing the maximum force limit in \eqref{eq:cv_force_max}. Note that linearizing any convex function about any trajectory, is always an inner approximation. Since the inner approximation is conservative, it is guaranteed that despite possibly being sub-optimal, all obtained solutions (if such solutions exist) are also feasible in the original non-convex problem. For more details, see Appendix \ref{ap:full_stat}.


\begin{figure}
 \centering
 \includegraphics[width=0.95\linewidth]{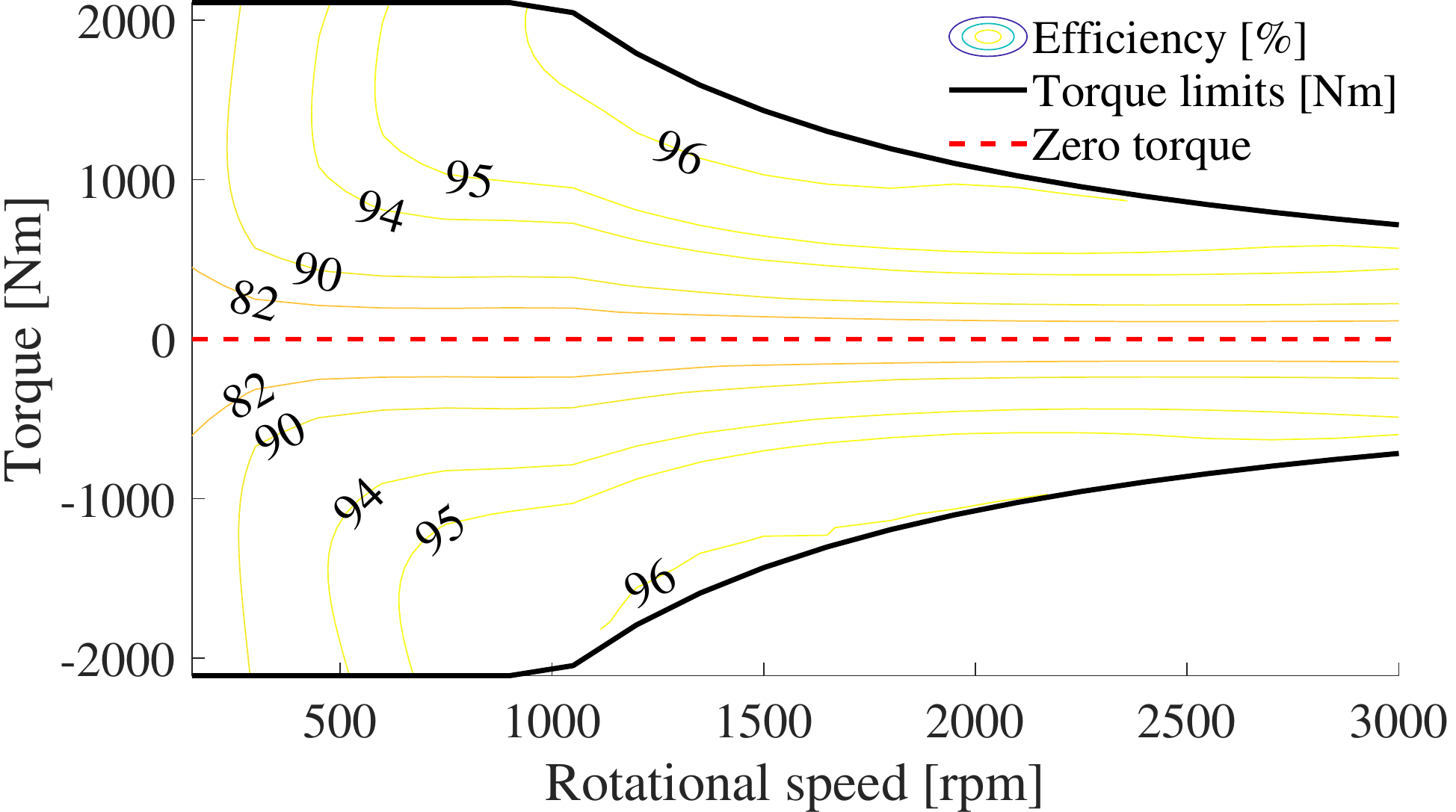}\vspace{-.3cm}
   \caption{\footnotesize 
   Steady-state  efficiency map of EM together with traction force limits delivered by EM.}
   \label{fig:eta_ev}
\end{figure}

\subsection{Fully electric vehicle}\label{subsec:on-ev}

In the fully electric powertrain, the EM converts electricity to mechanical power in motoring mode, whereas it converts mechanical power to electricity in generating mode of operation. In the generating mode, the energy is recuperated and stored in the electric battery, when decreasing kinetic energy by braking or decreasing potential energy while rolling downhill. Note that the electric powertrain is assumed to have a single-gear transmission system. 

For a given pair of rotational speed and torque, EM efficiency map is shown in Fig. \ref{fig:eta_ev}, using static internal electric battery power. In Fig. \ref{fig:eta_ev}, positive and negative torque regions correspond to the motoring and the generating modes of operation, respectively.

It is assumed that a single-gear transmission system conveys the power from the battery to the wheels. Therefore, there is no need for offline gear optimisation, i.e. $P_\gamma(v,F) = P_\tx{w}(v,F,\gamma)$.

The internal power drawn from the electric battery is approximated by the following expression
\begin{equation}
    \begin{split}
        P_\gamma(v,F) & \approx \tx{p}_\tx{m0} + \tx{p}_\tx{m1} v^3(s) + \tx{p}_\tx{m2} v(s)F(s)+\\
        & +\tx{p}_\tx{m3} v(s)F^2(s)
        \label{eq:exp_fv_ev}
    \end{split}
\end{equation}
with $\tx{p}_\tx{m0}, \tx{p}_\tx{m1}, \tx{p}_\tx{m2}, \tx{p}_\tx{m3} \geq 0$. Fig. \ref{fig:app_ev} demonstrates that the approximated model describes well the original internal battery power.


Using \eqref{eq:v2E}, \eqref{eq:F2a} and \eqref{eq:exp_fv_ev}, the stage cost \eqref{eq:p3_bottom} transforms into

\begin{equation}
\begin{split}
V_{\tx{EV}}(\cdot,\lambda_\tx{t}) & \approx \frac{\tx{c}_\tx{eg}(\tx{p}_\tx{m0}+\lambda_\tx{t}^*)\sqrt{\tx{m}}}{\sqrt{2 E(s)}} + \frac{2\tx{p}_\tx{m1}}{\tx{m}} E(s) + \tx{p}_\tx{m2}F(s) + \\
& + \tx{p}_\tx{m2}F^2(s) + \tx{w}_1 a^2(s)+\tx{w}_2 j^2(s).
\label{eq:ev_cost_stage2}
\end{split}
\end{equation}

The traction force limits, see Fig. \ref{fig:inn_app_ev}, are approximated by
\begin{align}
    & F_{\gamma\tx{min}}(E)\approx \max\Bigg\{\underline{F},\tx{x}_\tx{0} + \frac{\tx{x}_\tx{1} \sqrt{\tx{m}}}{\sqrt{2E(s)}}\Bigg\}
    \label{eq:ev_force_min}\\
    & F_{\gamma\tx{max}}(E)\approx \min\Bigg\{\overline{F},\tx{y}_\tx{0} + \frac{\tx{y}_\tx{1} \sqrt{\tx{m}}}{\sqrt{2E(s)}}\Bigg\}
    \label{eq:ev_force_max}
\end{align}
where $\underline{F}$ is constant minimum traction force. The coefficients $\tx{x}_\tx{0}$ and $\tx{x}_\tx{1}$, similar to the $\tx{y}_\tx{0}$ and $\tx{y}_\tx{1}$, are the solution of the linear program given in Appendix \ref{ap:lin_prg}. 

According to the signs of $\tx{x}_1\leq 0$ and $\tx{y}_1\geq 0$, the term $\tx{x}_1/\sqrt{E(s)}$ is a concave function and $\tx{y}_1/\sqrt{E(s)}$ is a convex function. Thus, the area between the two force limits \eqref{eq:ev_force_min} and \eqref{eq:ev_force_max} include a concave force set, which leads the problem \eqref{eq:p3} with the stage cost \eqref{eq:ev_cost_stage2} to be a non-convex nonlinear program. By linearizing the force limits, the problem \eqref{eq:p3} with the stage cost \eqref{eq:ev_cost_stage2} can be formulated as a convex SOCP, see Appendix \ref{ap:full_stat}. Note that linearizing any convex function about any trajectory, is always an inner approximation, and linearizing any concave function about any trajectory, results in an outer approximation. Furthermore, the approximations are conservative, therefore, all obtained solutions are inside the feasible force area, see Fig. \ref{fig:inn_app_ev}, and also feasible in the original non-convex problem.

\begin{figure}
 \centering
 \includegraphics[width=0.95\linewidth]{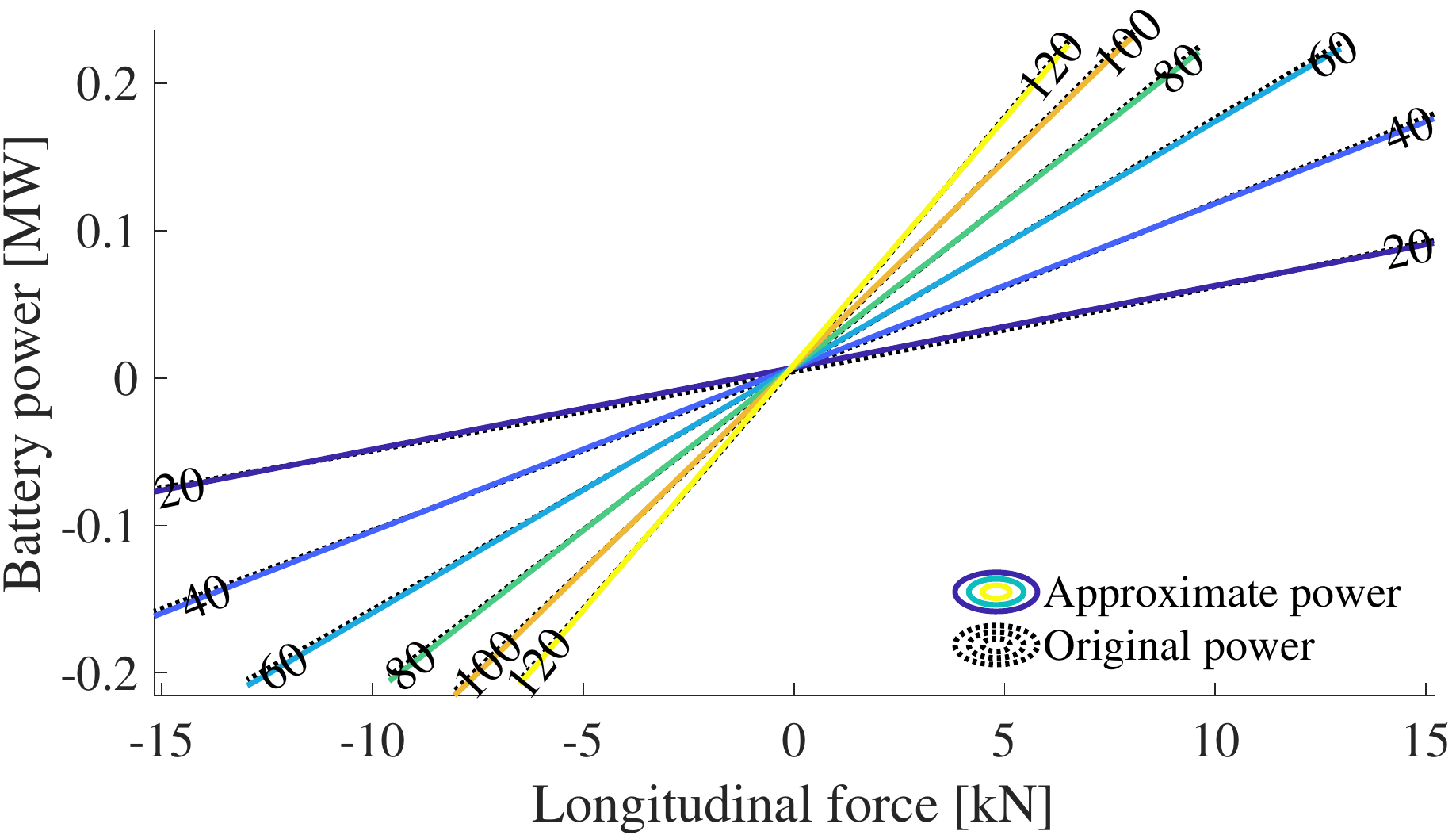}\vspace{-.3cm}
   \caption{\footnotesize Original and approximated internal power drawn from the electric battery for a given longitudinal velocity and traction force.}
   \label{fig:app_ev}
\end{figure}

\begin{figure}
 \centering
 \includegraphics[width=0.95\linewidth]{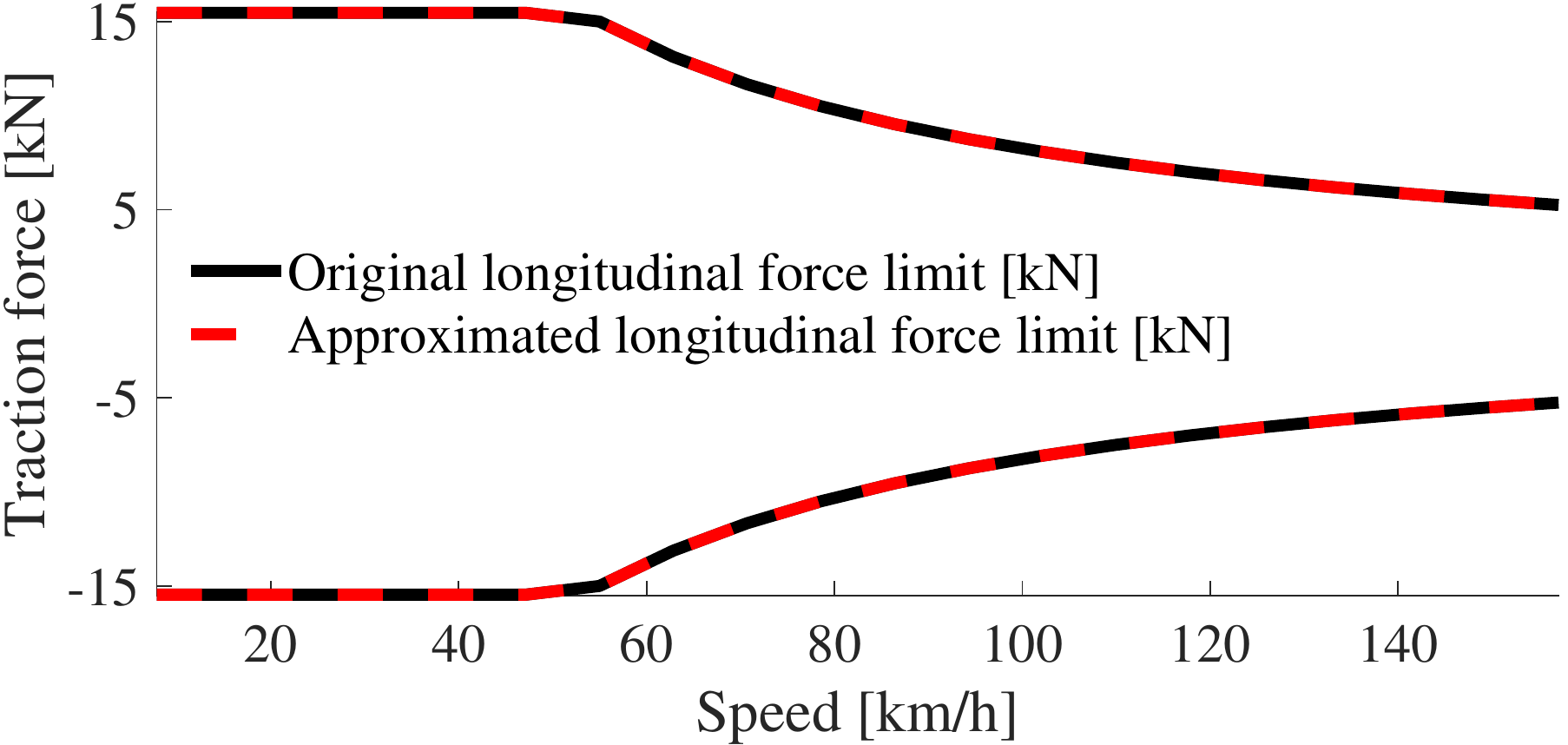}\vspace{-.3cm}
   \caption{\footnotesize Original and approximated traction force limits of the EM. 
   }
   \label{fig:inn_app_ev}
\end{figure}

\begin{table}
\begin{center}
\caption{Simulation parameters} 
\label{tab:par}
\begin{tabular}{l l}
\hline
Gravitational acceleration & $\tx{g}=  \SI{9.81}{m/s^2}$\\
Air density & $\uprho=\SI{1.29}{kg/m^3}$\\
Vehicle frontal area & $\tx{A}_\tx{f}=\SI{10}{m^2}$\\
Rolling resistance coefficient & $\tx{c}_\tx{r}=0.006$\\
Vehicle mass & $\tx{m}=\SI{40000}{kg}$ \\
Aerodynamic drag coefficient & $\tx{c}_\tx{d}=0.5$\\ 
Wheel radius & $\tx{r}_\tx{w}=\SI{0.50}{m}$ \\
Final gear ratio & $\tx{r}_\tx{fg}=3$ \\
Cruising set speed & $\tx{v}_\tx{cru}= \SI{80}{km/h}$\\
Route length & \SI{118}{km} 
\\
Number of samples & $N=400$ \\
Fuel cost & $\tx{c}_\tx{eg}^\tx{f}=\SI{1.51}{EUR/litre}$ \\
Electricity cost & $\tx{c}_\tx{eg}^\tx{e}=\SI{0.18}{EUR/kWh}$ \\
\hline
\end{tabular}
\end{center}
\end{table}

\section{Results}\label{sec:res}

In this paper, simulations are carried out for the CV and the EV over the \SI{118}{km} long road from S{\"o}dert{\"a}lje to Norrk{\"o}ping in Sweden, which is the same route as considered in \cite{eriksson16}. 
The problems \eqref{eq:p4_cv} and \eqref{eq:p4_ev} are discretized using the forward Euler method. For most of the simulation the sampling interval is kept at \SI{300}{m}, unless stated otherwise. The problems are solved in an SHMPC framework, i.e. ${\tx{s}_\tx{Hmax} \geq \tx{s}_\tx{f}}$, where 
travel time at the final position (end of the route) is upper bounded by $\tx{t}_\tx{f}$, using \eqref{eq:final_travel_time}. The simulation parameters are given in Table \ref{tab:par}.



Within the simulations we investigate: (1) how optimisation cost and optimal speed profile change for different discomfort penalties; (2) convergence properties of the algorithm; (3) computation time as a function of the number of samples in the horizon.

\begin{figure}[t!]
 \centering
 \subfigure[Fuel cost vs. RMS jerk.]{
 \includegraphics[width=\columnwidth]{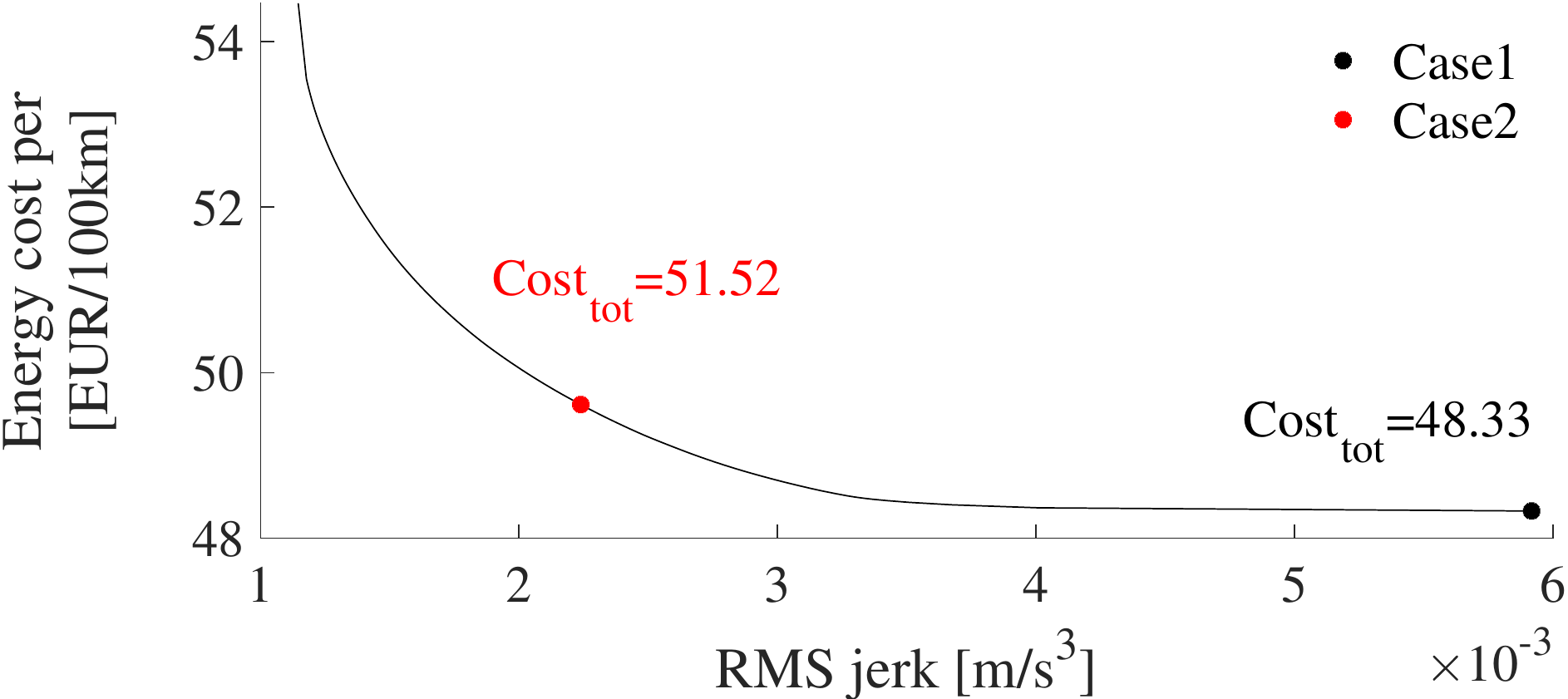}
 \label{fig:cost_rmsj_cv}
 }
 \subfigure[Electricity cost vs. RMS jerk.]{
 \includegraphics[width=0.95\linewidth]{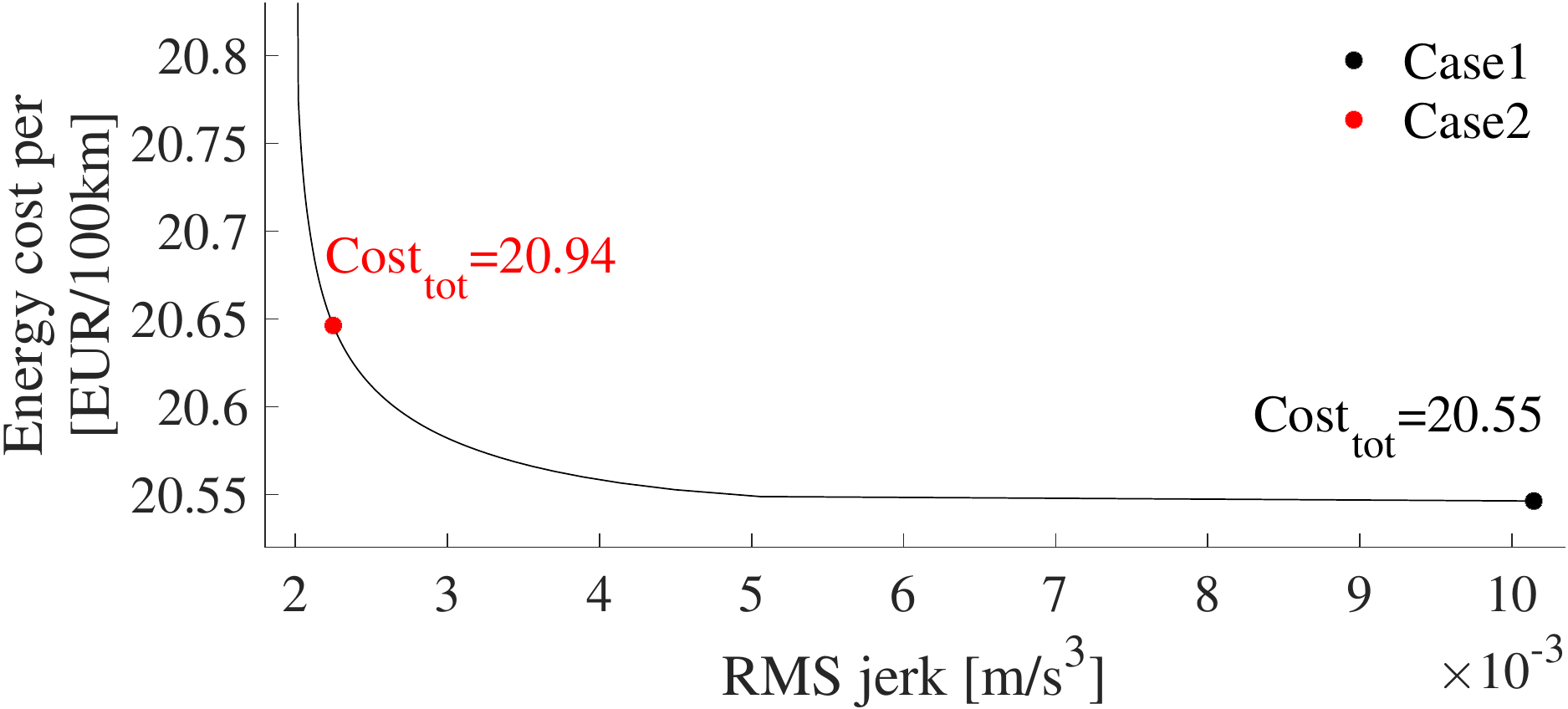}
 \label{fig:cost_rmsj_ev}
 }
 \caption{Energy cost investigation for different jerk penalty factors. For the large penalty factors, RMS jerk is saturated.}
 \label{fig:cost_rmsj}
\end{figure}

\begin{figure*}[t!]
 \centering
 \subfigure[Longitudinal velocity trajectories of CV.]{
 \includegraphics[width=0.48\linewidth]{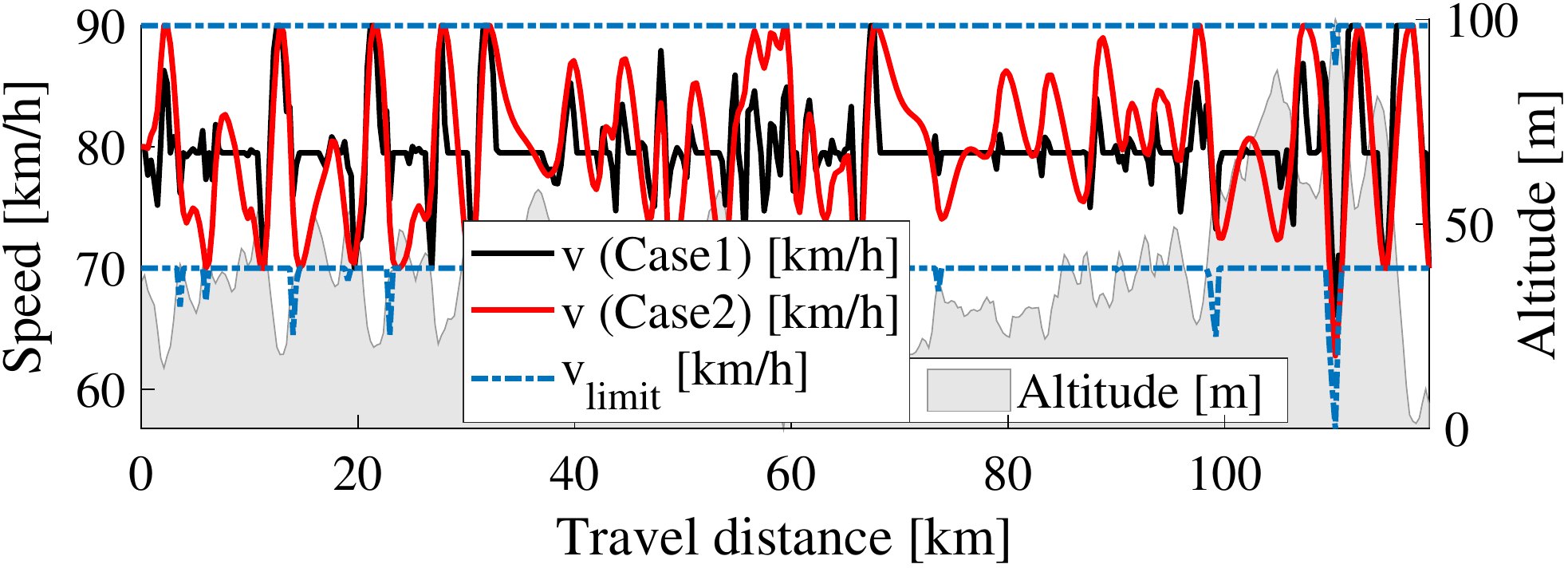}
 \label{fig:v_opt_cv}
 }
 \subfigure[Longitudinal velocity trajectories of EV.]{
 \includegraphics[width=0.48\linewidth]{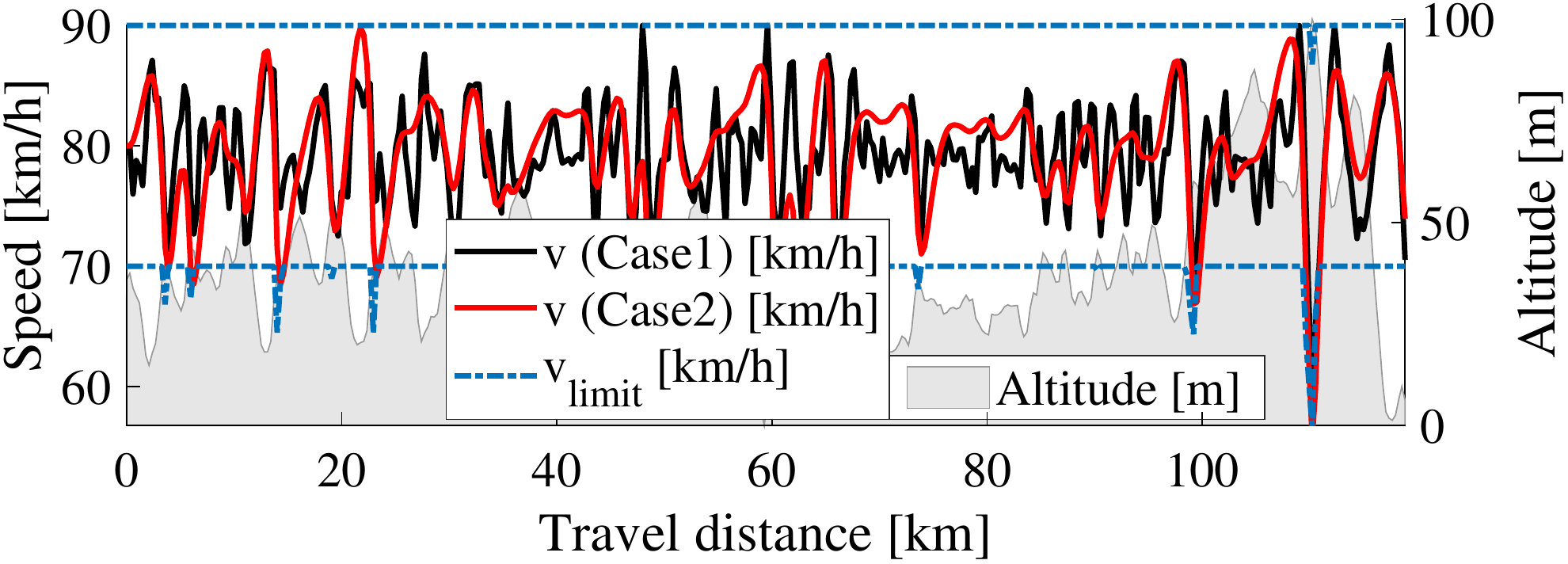}
 \label{fig:v_opt_ev}
 }
 \subfigure[Acceleration trajectories of CV.]{
 \includegraphics[width=0.48\linewidth]{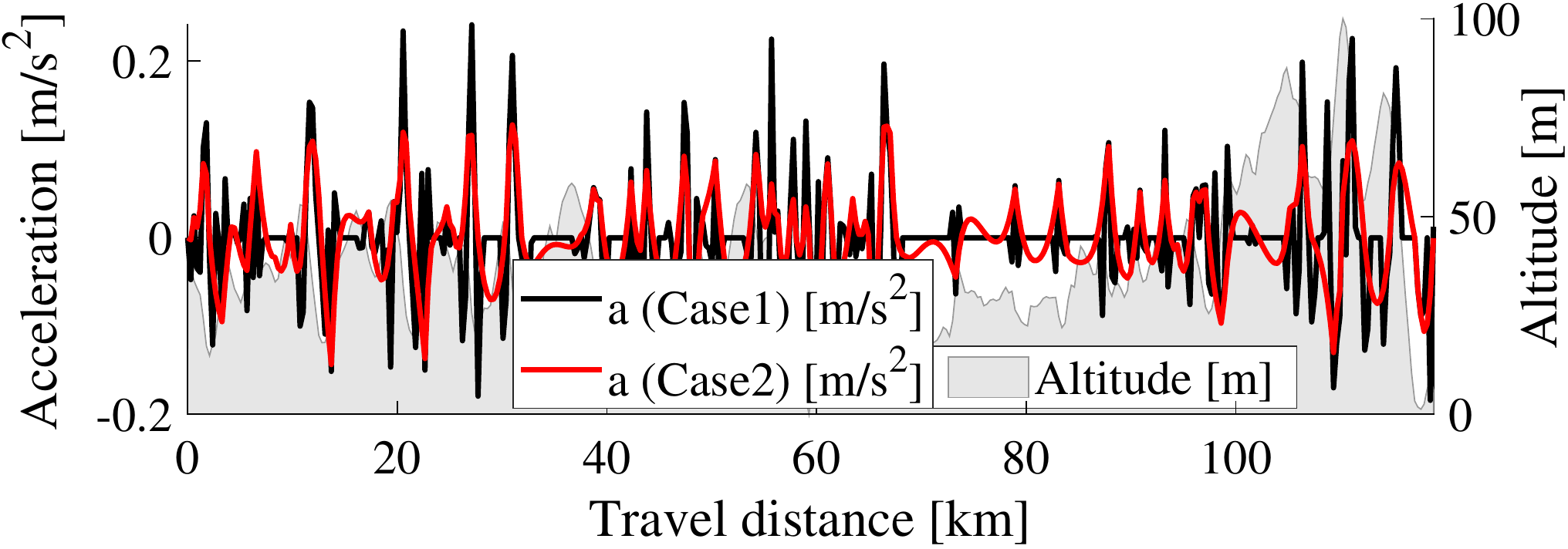}
 \label{fig:acc_opt_cv}
 }
 \subfigure[Acceleration trajectories of EV.]{
 \includegraphics[width=0.48\linewidth]{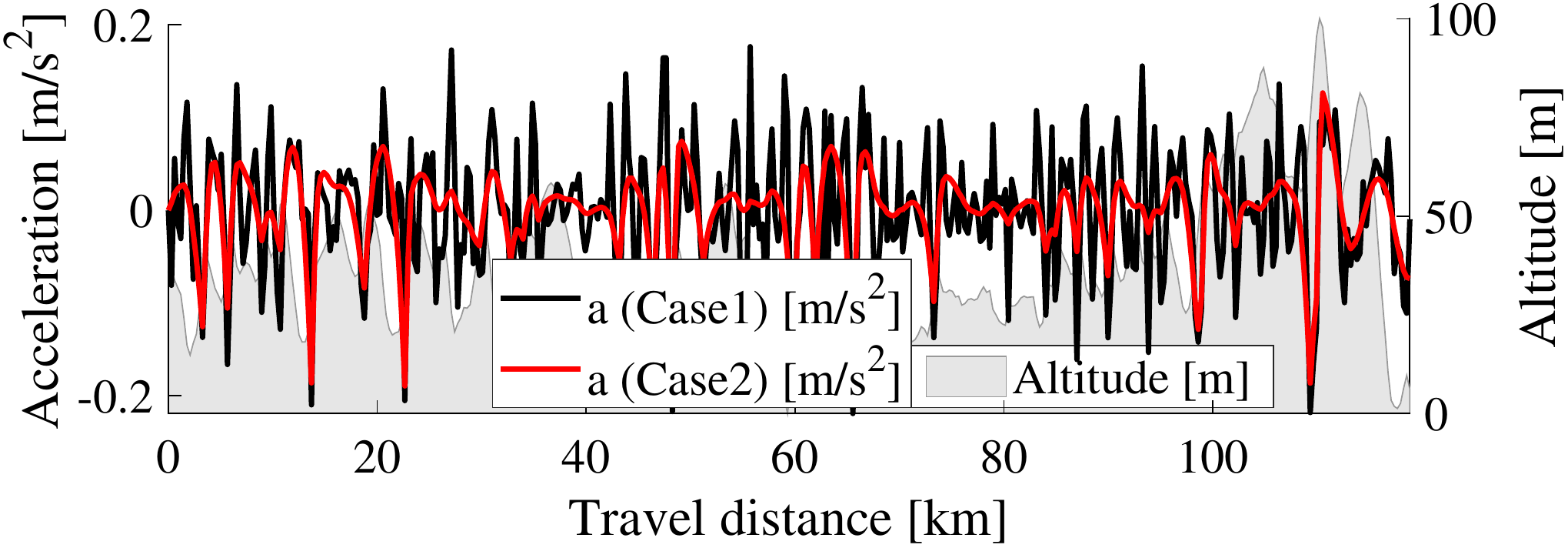}
 \label{fig:acc_opt_ev}
 }
 \subfigure[Jerk trajectories of CV.]{
 \includegraphics[width=0.48\linewidth]{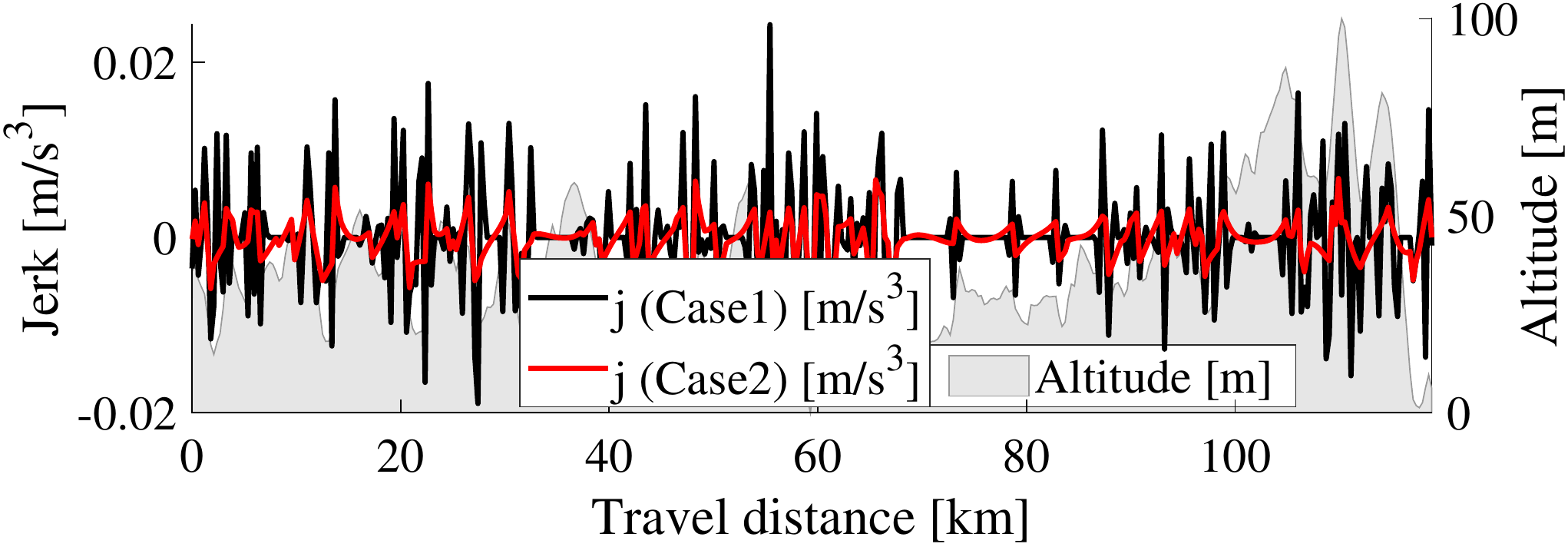}
 \label{fig:jrk_opt_cv}
 }
 \subfigure[Jerk trajectories of EV.]{
 \includegraphics[width=0.48\linewidth]{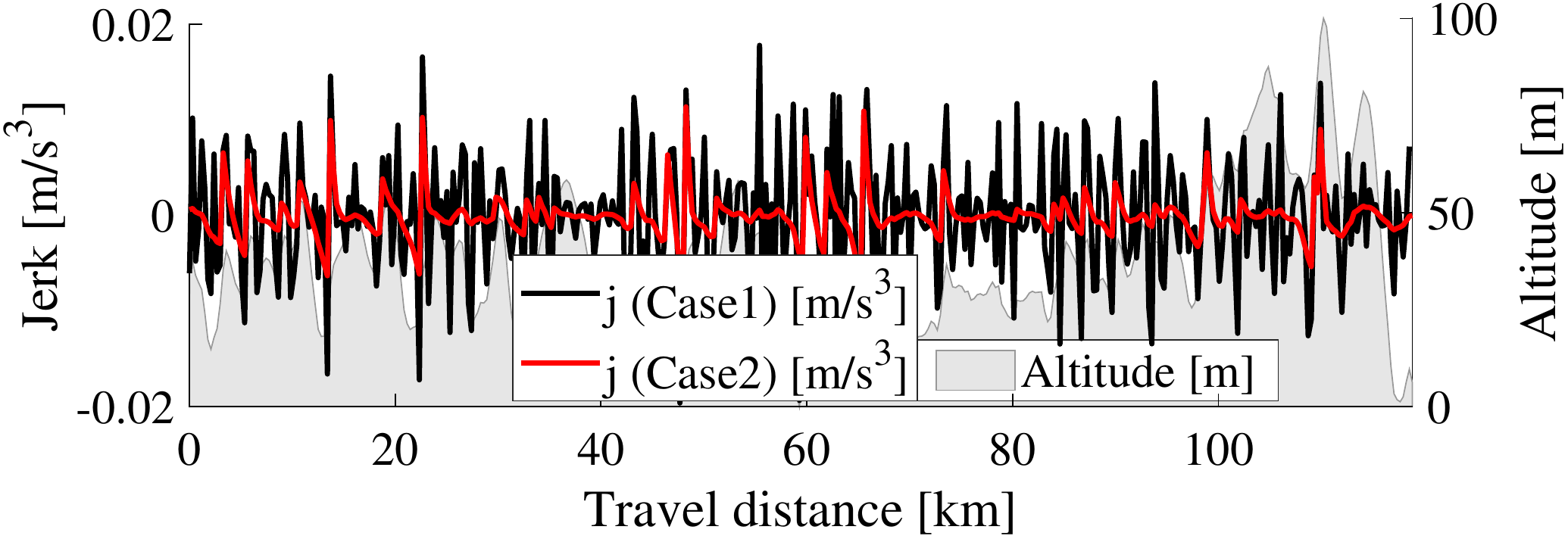}
 \label{fig:jrk_opt_ev}
 }
 \caption{Optimal longitudinal velocity, acceleration and jerk trajectories for CV and EV. Case 2, i.e. which corresponds to comfortable drive, provides smoother profile and more comfortable driving. Thus, the amplitude of fluctuating acceleration and jerk is decreased.}
 \label{fig:opt_sig}
\end{figure*}

\begin{figure}[t!]
\centering
\subfigure[Operating force-speed points of the CV and optimal gear as a contour map.]{
\includegraphics[width=0.95\linewidth]{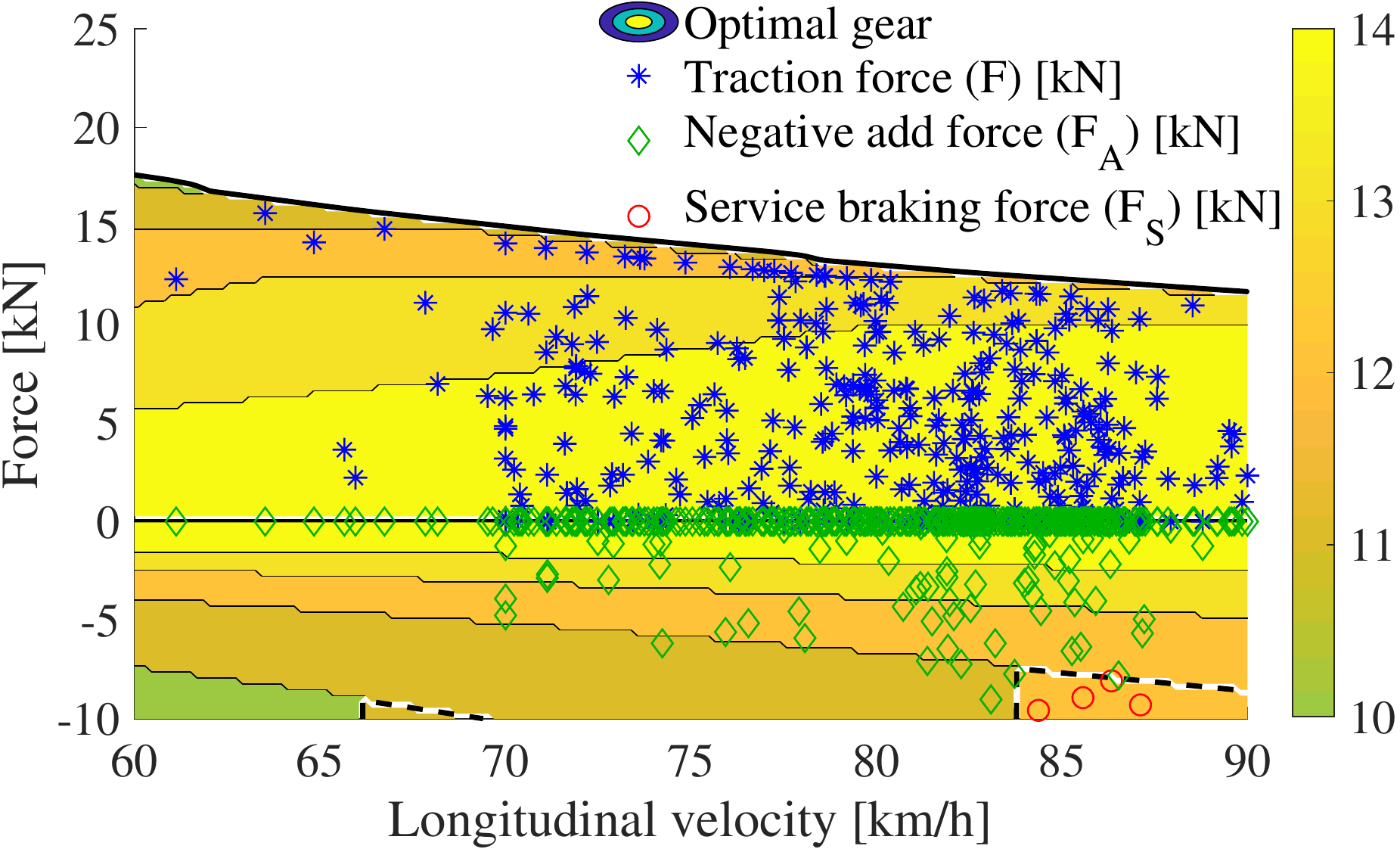}
\label{fig:opt_operate_cv}
}
\subfigure[Optimal force-speed points of the EV.]{
\includegraphics[width=0.95\linewidth]{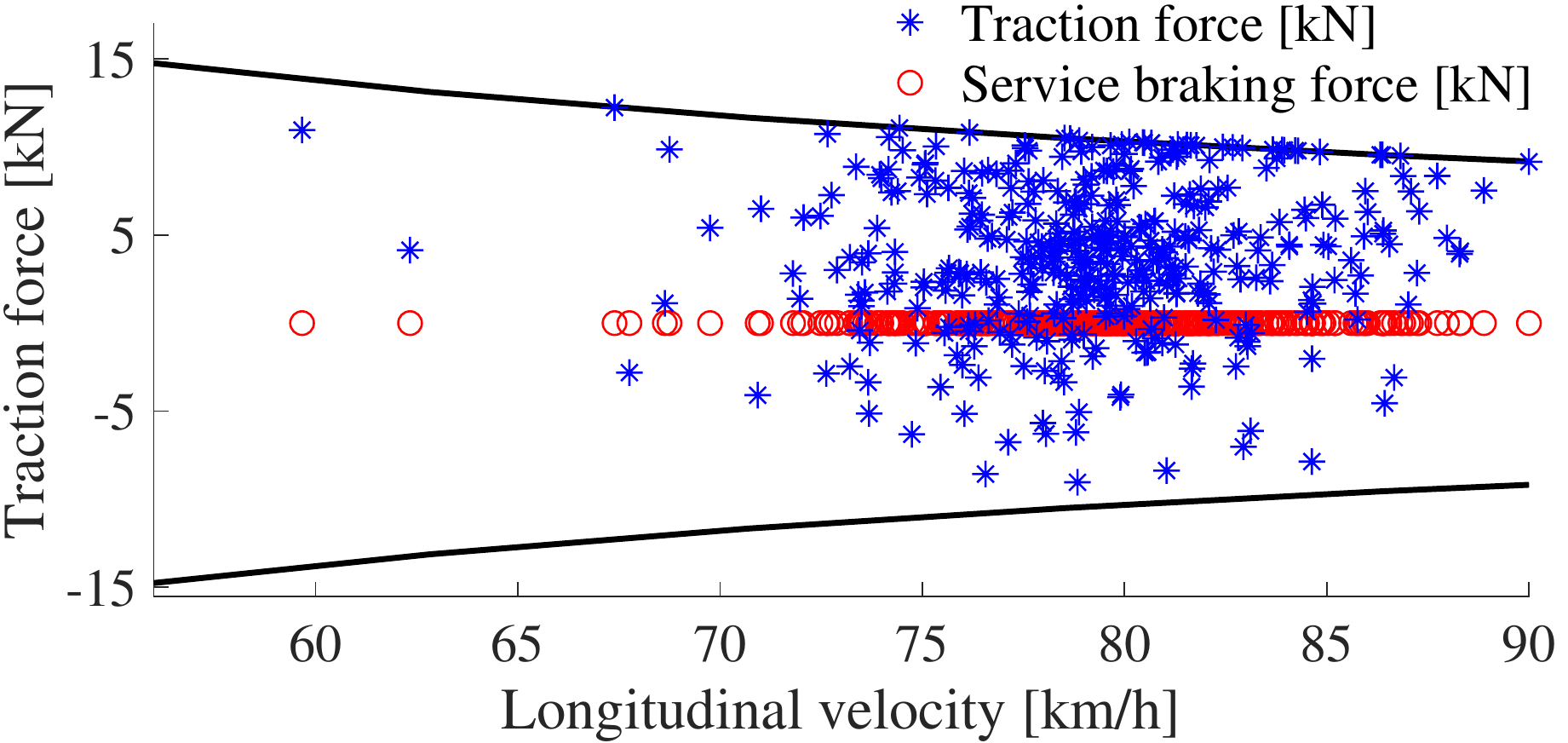}
\label{fig:opt_operate_ev}
}
\caption{Optimal longitudinal forces vs. vehicle speed for Case 2, i.e. when jerk is penalised.}
\label{fig:opt_operate}
\end{figure}

\begin{figure}[t!]
 \centering
 \subfigure[Optimal gear trajectory.]{
 \includegraphics[width=0.95\linewidth]{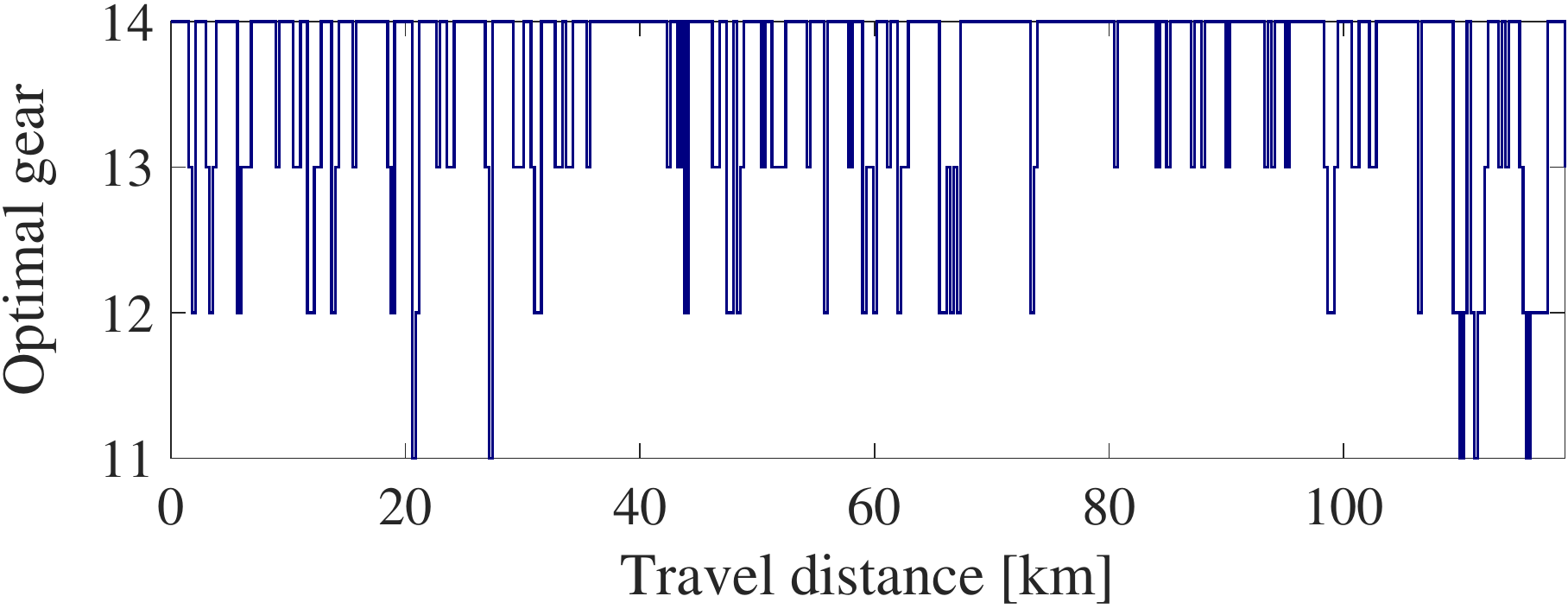}
 \label{fig:opt_gear_trj}
 }
 \subfigure[Optimal gear occurrence.]{
 \includegraphics[width=0.95\linewidth]{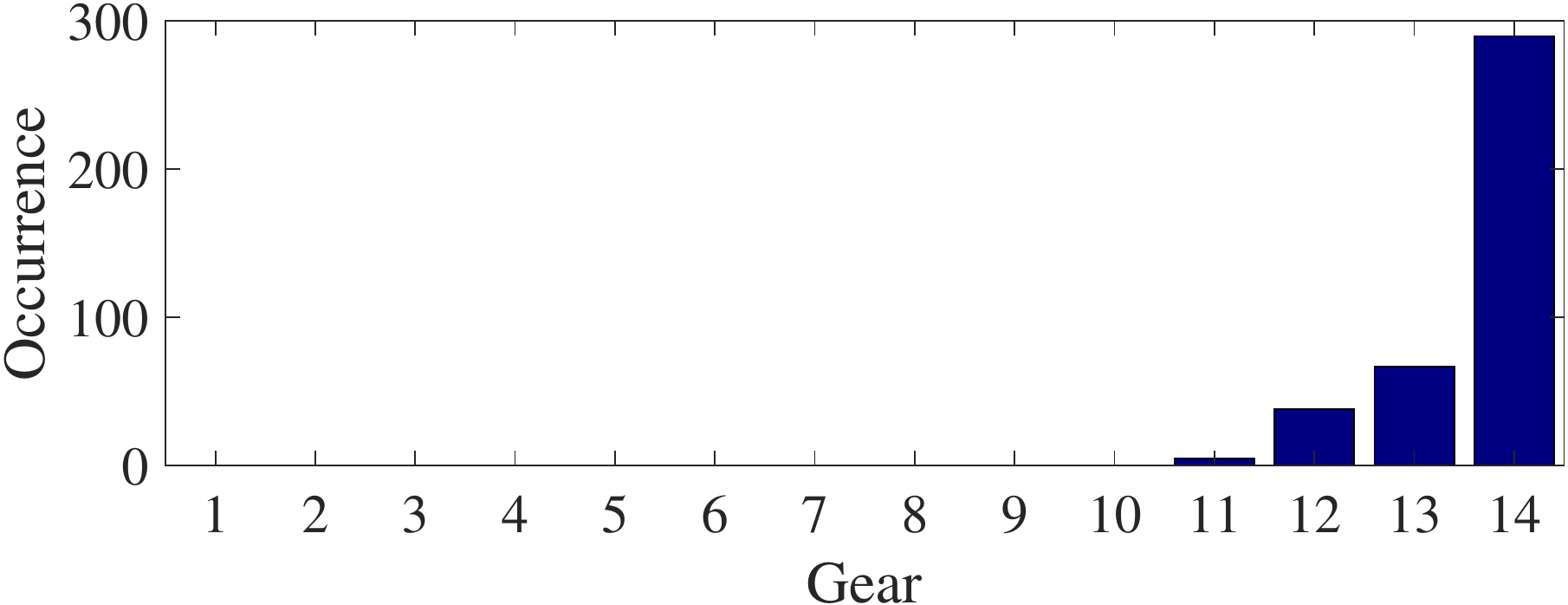}
 \label{fig:opt_gear_dis}
 }
 \caption{Optimal gear profiles of CV for Case 2, i.e. which corresponds to comfortable drive. The most frequent selected gear is $\gamma=14$.}
 \label{fig:opt_gear_trjdis}
\end{figure}

\subsection{Energy consumption vs. drivability}\label{subsec:res2}
To study the cost components, i.e. energy cost and the cost due to penalising discomfort, we compare three case studies: $\tx{Case}^\tx{hg}$ corresponds to a case with the heuristic guess for longitudinal velocity, $v_{\tx{hg}}$. For this case, the stage costs, \eqref{eq:p4_cv} and \eqref{eq:p4_ev} are calculated using \eqref{eq:v2E}, \eqref{eq:a_dyn_s} and \eqref{eq:j_dyn_s}.
In Case 1,  i.e. performance drive, the jerk penalty term in \eqref{eq:p4_cv} and \eqref{eq:p4_ev} is kept to zero; and in Case 2, i.e. comfortable drive, non-zero jerk penalty factor in \eqref{eq:p4_cv} and \eqref{eq:p4_ev} leads to smooth 
velocity. As an index to measure drivability, the root mean square (RMS) value of jerk
\begin{align}
    j_{\tx{RMS}}=\sqrt{\frac{1}{\tx{s}_\tx{f}}\int_{0}^{\tx{s}_\tx{f}}j^2(s)\tx{d}s}
\end{align}
is used. Note that we have observed the smooth speed profile could be achieved by only penalising jerk, thus the penalty coefficient on the acceleration, $\tx{w}_1$, is always kept to be zero for all three cases.

There is a trade-off between the energy cost and comfort, i.e. lower values of RMS jerk yield higher energy cost, see Fig. \ref{fig:cost_rmsj_cv} and Fig. \ref{fig:cost_rmsj_ev} for such trade-off for the CV and the EV respectively. Thus, vehicle manufacturers have wide range of choice to customise the vehicle's performance for a desired energy use and comfort. Note that RMS jerk saturates for large jerk penalty factors. Hereafter, the jerk penalty term in Case 2 is selected in a way that the RMS jerk is equal to \SI{0.0022}{m/s^3} for the CV and the EV.

Optimal longitudinal velocity, acceleration and jerk profiles of Case 1 and Case 2 for the CV and the EV are demonstrated in Fig. \ref{fig:opt_sig}. The velocity profiles without discomfort penalty, i.e. Case 1, are saw-tooth shaped and leads to more aggressive way of driving, however, the latter case provides smoother and more comfortable driving, see Fig. \ref{fig:v_opt_cv} and Fig. \ref{fig:v_opt_ev}. Note that in addition to the RMS jerk, the RMS acceleration is also reduced in Case 2 compared to Case 1 for the CV and the EV, whereas the acceleration is not penalised in either cases, see Fig. \ref{fig:acc_opt_cv}, Fig. \ref{fig:acc_opt_ev}, \ref{fig:jrk_opt_cv} and Fig. \ref{fig:jrk_opt_ev}. 

Optimal traction and braking force points for the Case 2, i.e. comfortable drive, of CV and EV are shown in Fig. \ref{fig:opt_operate}. Also, according to the optimal gear map in Fig. \ref{fig:gear_cv}, for a pair of total force and longitudinal velocity, the optimal gear is chosen. The optimised gear trajectory and distribution are shown in Fig. \ref{fig:opt_gear_trjdis}, where the most frequently selected gear is $\gamma=14$. We have observed similar results for Case 1 as well.

The cost results of the whole driving mission and their corresponding RMS jerk values for all three case studies of the CV and the EV are given in Table \ref{tab:cvev}. 

For the CV, the most fuel-efficient case is Case 1. There is a benefit of $15.71\%$ to optimize the velocity profile compared to the $\tx{Case}^\tx{hg}$, whereas the discomfort of the performance drive is accepted. Furthermore, the results show $10.14\%$ reduction in total cost of Case 2 compared to the $\tx{Case}^\tx{hg}$, despite having $2.66\%$ increase in fuel consumption compared to Case 1. As it has been expected, the proposed algorithm minimises the braking at the pads, i.e. the braking in Case 1 and Case 2 is significantly reduced compared to $\tx{Case}^\tx{hg}$.

For the EV, Case 1 provides $5.20\%$ reduction of the total energy cost compared to $\tx{Case}^\tx{hg}$ and the total cost benefit of Case 2 is $3.11\%$ compared to $\tx{Case}^\tx{hg}$. The comfortable drive, i.e. Case 2, leads to $0.49\%$ increase in electricity usage compared to the performance drive, i.e. Case 1. Note that the electricity cost in Case 2 is slightly worst than that in $\tx{Case}^\tx{hg}$, i.e. $0.04\%$, which implies that the heuristic guess is a proper guess. However, the RMS jerk in Case 2 is reduced by $40\%$ compared to $\tx{Case}^\tx{hg}$, i.e. the RMS jerk is reduced from \SI{0.0037}{m/s^3} to \SI{0.0022}{m/s^3}.   



\begin{table}[t]
\caption{Simulation results, energy consumption vs. drivability} \label{tab:cvev}
  \centering
  \setlength\tabcolsep{6pt}
 \begin{tabular}{|c| c c c |} 
 \hline
 \multicolumn{4}{| c |}{\textbf{CV}}\\
 \hline
 \begin{tabular}{c}Variable\end{tabular} & \begin{tabular}{c}$\tx{Case}^\tx{hg}$\end{tabular}& \begin{tabular}{c}Case 1\end{tabular}&
 \begin{tabular}{c}Case 2\end{tabular}\\
 \hline
$\tx{Fuel cost}$ [EUR] & 65.23 & 57.44 & 58.97 \\
$\tx{Drivability cost}$ [EUR] & 2.92 & 0 & 2.27 \\
\hline
$\tx{Total cost}$ [EUR] & 68.15 & 57.44 & 61.24 \\
Improvement [\si{\percent}]& - & 15.71 & 10.14\\
\hline 
$j_{\tx{RMS}}$ [\SI{}{m/s^3}] & 0.0024 & 0.0059 & 0.0022 \\
$||F_\tx{brk}||$ [\SI{}{kN}] & 45.30 & 24.90 & 25.10 \\
\hline
\hline
\multicolumn{4}{| c |}{\textbf{EV}}\\
 \hline
 \begin{tabular}{c}Variable\end{tabular} & \begin{tabular}{c}$\tx{Case}^\tx{hg}$\end{tabular}& \begin{tabular}{c}Case 1\end{tabular}&
 \begin{tabular}{c}Case 2\end{tabular}\\
 \hline
$\tx{Electricity cost}$ [EUR] & 24.54 & 24.42 & 24.55 \\
$\tx{Drivability cost}$ [EUR] & 1.22 & 0 & 0.41 \\
\hline
$\tx{Total cost}$ [EUR] & 25.76 & 24.42 & 24.96 \\ 
Improvement [\si{\percent}] & - & 5.20 &  3.10\\
\hline
$j_{\tx{RMS}}$ [\SI{}{m/s^3}] & 0.0037 & 0.012 & 0.0022 \\
$||F_\tx{brk}+\min(F, 0)||$ [\SI{}{kN}] & 44.90 & 25.50 & 35.70 \\
\hline
\end{tabular}
\end{table}

\begin{figure}[t]
\centering
\begin{tikzpicture}
\node[rotate=0] at(0,0){\includegraphics[width=0.95\linewidth]{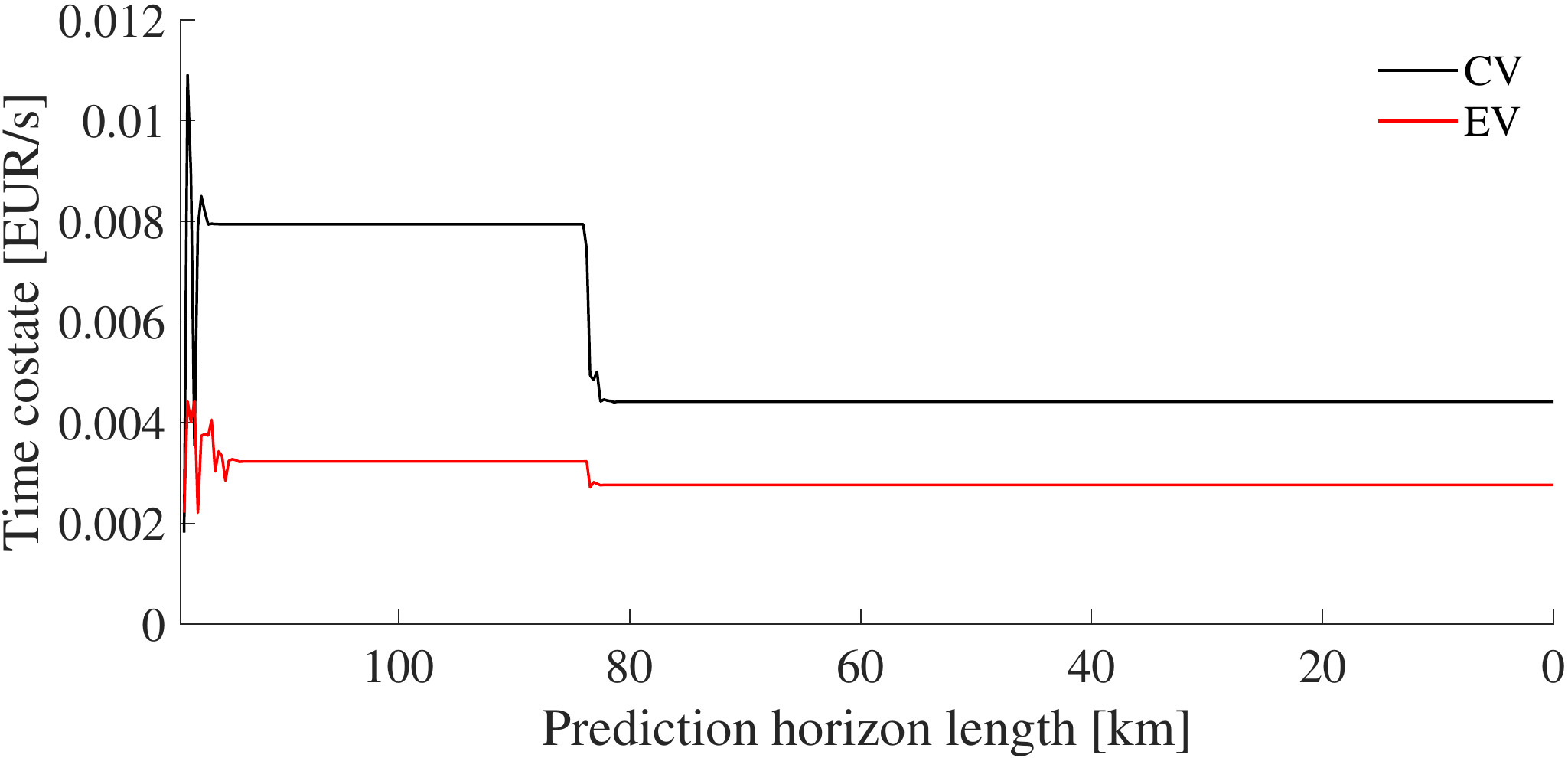}};
\end{tikzpicture}
\vspace{-.4cm}
\caption{\footnotesize Travel time costate vs. prediction horizon length. The costate converges after few MPC updates, even after disturbance is introduced (at horizon length of \SI{85}{km}) by suddenly increasing maximum travel time, e.g. due to traffic congestion. 
}
\label{fig:res3A_cvev}
\end{figure}

\begin{figure}[t]
\centering
\begin{tikzpicture}
\node[rotate=0] at(0,0){\includegraphics[width=0.95\linewidth]{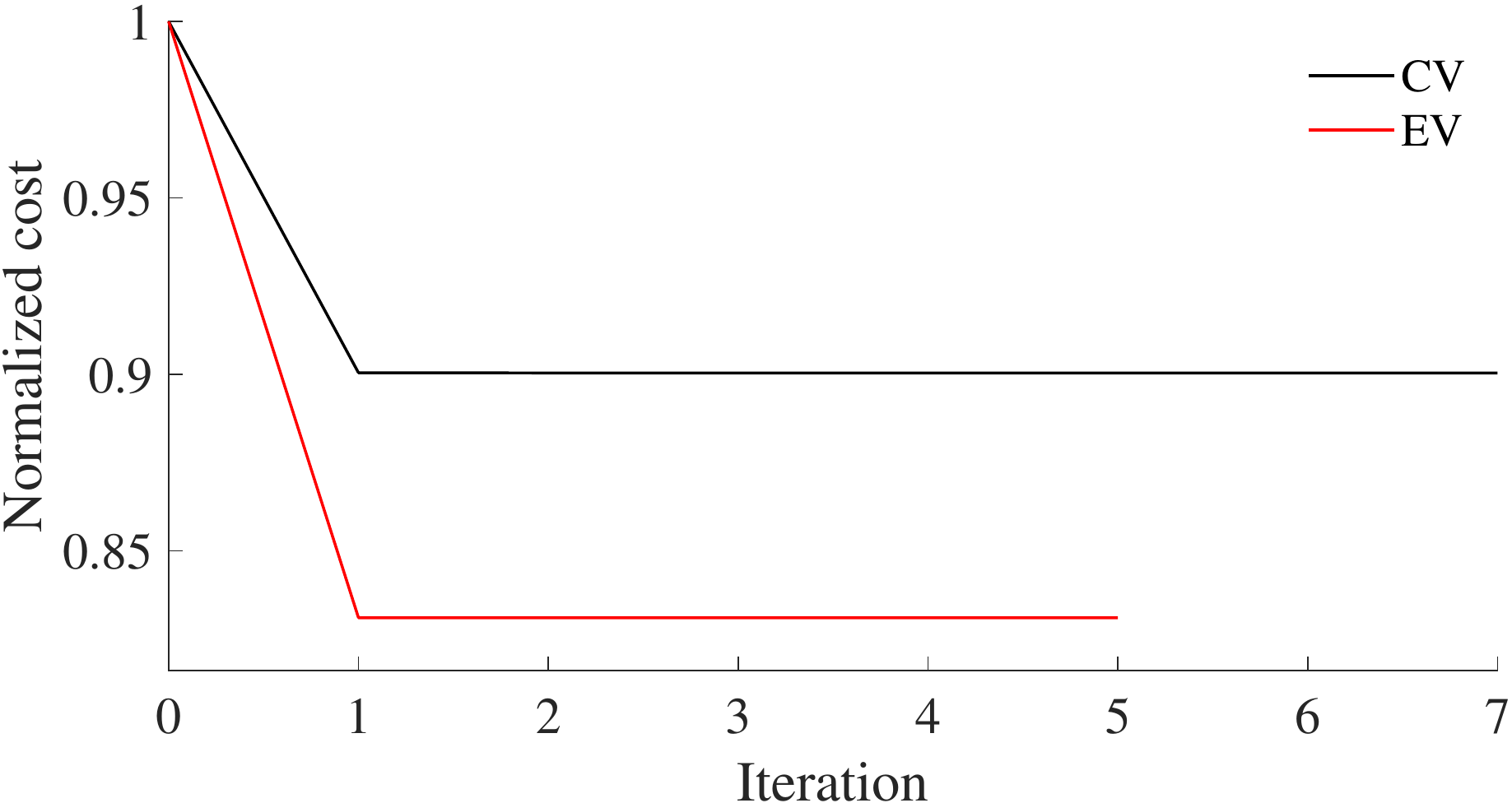}};
\end{tikzpicture}
\vspace{-.4cm}
\caption{\footnotesize SQP convergence profile. The cost value drops to within 0.4\% from optimum value in the first iteration. In iteration 0 the cost value is calculated when the vehicle is driving with the initial estimated trajectory, $v_{\tx{hg}}$.}
\label{fig:res3C_cvev}
\end{figure}

\subsection{Algorithm convergence}\label{subsec:res3}
The convergence curve of the time costate versus shrinking prediction horizon length is shown in Fig. \ref{fig:res3A_cvev}. According to the algorithm given in Appendix \ref{ap:lambda_opt}, the time costate is updated once per each MPC stage rather than waiting for the full costate convergence. It can be observed that after few initial MPC stages, the time costate converges to its optimum value. The disturbance rejection properties of the algorithm are verified in Fig. \ref{fig:res3A_cvev}. At the prediction horizon of \SI{85}{km}, maximum travel time changes due to e.g. traffic congestion. It can be seen in Fig. \ref{fig:res3A_cvev} that the travel time costate converges to its new value, which leads the vehicle to arrive to the final position within the updated maximum travel time.       

The convergence profile of the SQP algorithm is depicted in Fig. \ref{fig:res3C_cvev}, where the algorithm converges to an optimum obtained by solving \eqref{eq:p3} in $7$ iterations for CV and $5$ iterations for EV. However, the cost value drops to within $0.4\%$ from the optimum value in the first iteration. We exploit this behaviour through RTI in SHMPC framework, where only one QP is solved in each MPC update rather than waiting for the full SQP convergence, since the cost value in the first iteration is very close to the local optimum. Note that the cost value in iteration $0$ is calculated when the vehicle is driving with the initial estimated trajectory, $v_{\tx{hg}}$.

\begin{figure}
 \centering
 \includegraphics[width=0.95\linewidth]{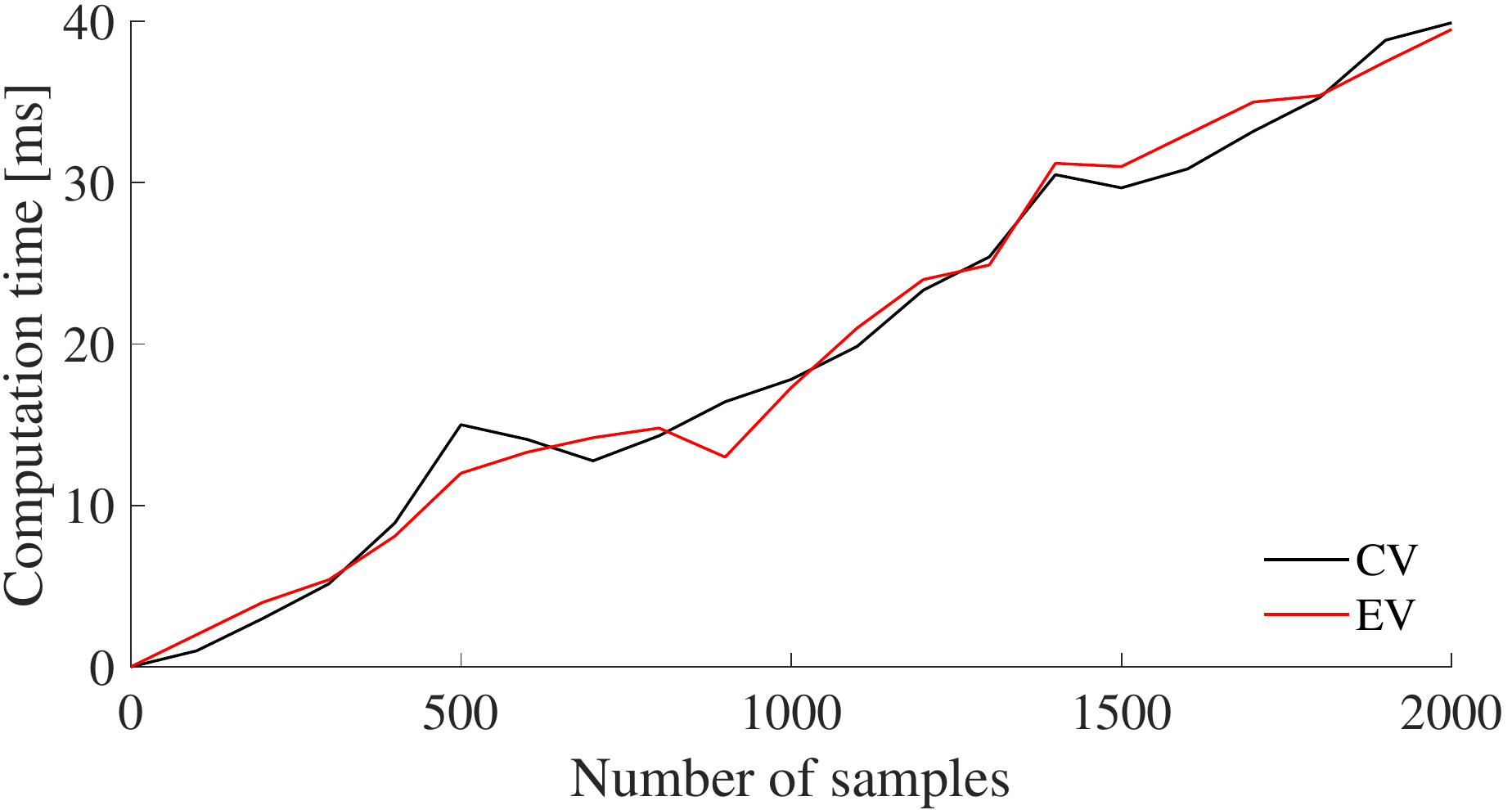}\vspace{-.3cm}
   \caption{\footnotesize Computation time vs. prediction horizon length using HPIPM for various resolutions of the prediction horizon. The computation time increases linearly with the number of samples.}
   \label{fig:ct_cvev}
\end{figure}

\subsection{Computation time}\label{subsec:res4}

The computation time profile for various sampling intervals is depicted in Fig. \ref{fig:ct_cvev} using HPIPM, where the entire route, \SI{118}{km}, is considered as the prediction horizon. The optimisation was run on a laptop PC with 6600K CPU at $2.81$GHz and $16$GB  RAM. The trend is that as the number of samples increases, the computation time also increases. For real-time applications, it is preferable to have small sampling interval, however the information on the topography should not be lost. In subsections \ref{subsec:res2} and \ref{subsec:res3}, the number of samples is kept to $400$ and the corresponding computation time for solving the problems \eqref{eq:p4_cv} and \eqref{eq:p4_ev} is less than $10$ms, which is considerably low value for a horizon of \SI{118}{km}.

\section{Conclusion}\label{sec:conclusion}
In this paper a computationally efficient algorithm is proposed for eco-driving over long look-ahead horizons. To this end, a bi-level program is formulated, where integer variable, i.e. gear, is decoupled from the real-valued variables. In the bottom level, the optimal gear map is derived in a way that the total energy consumption is minimised. In the top level, the remaining online implementable NLP is formulated. To provide more comfortable way of driving, acceleration and jerk of the vehicle are penalised in the top level's objective. In the NLP, the dynamics on travel time is adjoined to the objective function, using the necessary PMP conditions for optimality, since: 1) the Hamiltonian is not an explicit function of the travel time; 2) the travel time is strictly monotonically increasing function; and 3) the constraint on final time may activate at the final instant. The NLP is solved by applying RIT SQP scheme in MPC framework, i.e. the time costate and the linearization trajectory are updated once per each MPC update. The proposed algorithm is applied to a CV and an EV using SHMPC framework.

According to the simulation results, there is a trade-off between cost and comfort, i.e. driving comfortably is more expensive compared to the performance drive. The energy increase because of penalising the driver's discomfort is $2.66\%$ and $0.49\%$ for the CV and EV, respectively, where the RMS jerk is kept to \SI{0.0022}{m/s^3}. Also, by using the proposed algorithm, the total cost is reduced up to $15.71\%$ and $5.20\%$  for the CV and EV, respectively, compared to the case with initial velocity estimation. The computation time for the horizon of \SI{118}{km} is \SI{10}{ms}, the sampling interval is equal to \SI{300}{m}. For on-line applications, the small computation time can enhance the optimality, since the suggested optimal state of vehicle can be updated more frequently. Also, in off-line analysis the small computation time can be applied to multi-path problems, where the optimal path of the driving vehicle in terms of energy consumption can be obtained within a small amount of time. The presented algorithm in this paper can also be applied to HEVs as well, where the battery discharge trajectory is generated by the mission manager and delivered to lower control layers to charge depleting or charge sustaining operation. 


\appendices
\section{Newton method for finding optimal time costate}\label{ap:lambda_opt}
In this paper, a modified Newton method is applied to find the $\lambda_\tx{t}^*$. Let
\begin{align}
   f(\lambda_\tx{t}|\zeta)=t^*(\lambda_\tx{t},s_\tx{H}|\zeta) - t_\tx{H}(\zeta). \label{eq:f}
\end{align}

The rule for updating $\lambda_\tx{t}$ is 
\begin{align}
     \lambda_\tx{t}(\zeta^{+})=\lambda_\tx{t}(\zeta)-\frac{f(\lambda_\tx{t}|\zeta)}{\Tilde{f}^{'}(\lambda_\tx{t}|\zeta)}
     \label{eq:newton}
\end{align}
with
\begin{align}
    &\Tilde{f}^{'}(\lambda_\tx{t}|\zeta)=\min_{\lambda_\tx{t}}\bigg\{f^{'}(\lambda_\tx{t}|\zeta),f_{\max}^{'}\bigg\},\\
    &f^{'}(\lambda_\tx{t}|\zeta)=\frac{f(\lambda_\tx{t}|\zeta)-f(\lambda_\tx{t}|\zeta^{+})}{\lambda_\tx{t}(\zeta)-\lambda_\tx{t}(\zeta^{+})},\\
    &f_{\max}^{'}=\frac{f_{\max}-f_{\min}}{\lambda_\tx{t}^{\min}-\lambda_\tx{t}^{\max}}
\end{align}
where $\lambda_\tx{t}^{\min}=0 \tx{EUR}/s$ is the minimum and $\lambda_\tx{t}^{\max}$ is the maximum time costate. Also, $f_{\min}$ and $f_{\max}$ are
\begin{align}
    &f_{\min}=t^*(\lambda_\tx{t}^{\max},s_\tx{H}|\zeta) - t_\tx{H}(\zeta), \label{eq:fminNewton}\\
    &f_{\max}=t^*(\lambda_\tx{t}^{\min},s_\tx{H}|\zeta) - t_\tx{H}(\zeta). \label{eq:fmaxNewton}
\end{align}

To speed up the convergence to $\lambda_\tx{t}^*$ in \eqref{eq:newton}, it is possible to warm start the algorithm by initialising $\lambda_\tx{t}$ at two consecutive instances $\zeta=0$ and $\zeta=0^{+}$, as 
\begin{align}
    &\lambda_\tx{t}(0)=\lambda_\tx{t}^{\min}-\frac{f_{\max}}{f_{\max}^{'}}\\
    &\lambda_\tx{t}(0^{+})=\lambda_\tx{t}(0)-\frac{f(\lambda_\tx{t}|0)}{f_{\max}^{'}}.
\end{align} 
where $\lambda_\tx{t}(0)$ is simply the intersection point of $f(\lambda_\tx{t}|0)=0$ with a line connecting the two points $(\lambda_\tx{t}^{\min},f_{\max})$ and $(\lambda_\tx{t}^{\max},f_{\min})$. 

\section{Inner approximation of traction force limits}\label{ap:lin_prg}
To approximate the force limits 
as inner approximations of the original nonlinear and non-smooth limits, a linear program is solved as:
\begin{align}
& J=\min_{x} \left( f^Tx \right) \nonumber\\
& \text{subject to} \nonumber\\
& Ax \leq b
\end{align}
such that the area between actual force limits and their approximations is minimised. Therefore, the area between the approximated force limit and the line $F=0$ is maximised. To this end, for the minimum force limit
\begin{align}
          J=\min_{x}\int_{\tx{v}_0}^{\tx{v}_\tx{max}}(\tx{x}_0+\frac{\tx{x}_1}{v})\tx{d}v
\end{align}
and for the maximum force limit
\begin{align}
          J=\min_{x}\int_{\tx{v}_0}^{\tx{v}_\tx{max}}-(\tx{y}_0+\frac{\tx{y}_1}{v})\tx{d}v.
\end{align}
Thus,
\begin{align}
          A = \begin{bmatrix}
           1 &
           \frac{1}{v} 
         \end{bmatrix},
\end{align}
for the minimum force limit, $f,b,x$ are defined as
\begin{align}
& f = -\begin{bmatrix}
           \tx{v}_\tx{max}-\tx{v}_0 \\
           \ln(\tx{v}_\tx{max}) - \ln(\tx{v}_0) \\
         \end{bmatrix}, \ b=F_{\gamma\tx{min}}(v), \ x = \begin{bmatrix}
           \tx{x}_0 \\
           \tx{x}_1 
         \end{bmatrix}
\end{align}
and for the maximum force limit as
\begin{align}
& f = \begin{bmatrix}
           \tx{v}_\tx{max}-\tx{v}_0 \\
           \ln(\tx{v}_\tx{max}) - \ln(\tx{v}_0) \\
         \end{bmatrix}, \ b=F_{\gamma\tx{max}}(v), \quad x = \begin{bmatrix}
           \tx{y}_0 \\
           \tx{y}_1 
         \end{bmatrix}.
\end{align}

The vehicle speed, $v$, is allowed to vary between two limits 
\begin{align}
& v \in [\tx{v}_0, \tx{v}_\tx{max}]
\nonumber
\end{align}
where for CV $\tx{v}_0$=\SI{8}{km/h} and for EV $\tx{v}_0$=\SI{55}{km/h}, and $\tx{v}_\tx{max}$ is the maximum reachable speed by the vehicle. In this formulation, the idea is to minimize the area between the original force limit and the inner approximation.
\section{Full statement of convex optimal energy consumption program}\label{ap:full_stat}

Here, the full statement of convex optimal energy consumption problem is given for CV and EV case studies. To this end, the nonlinear term $f(E)=1/\sqrt{E(s)}$ in \eqref{eq:cv_force_max} is linearized about a trajectory $\hat E(s)$,
\begin{align} 
    f^{\tx{lin}}(E,\hat E) \approx f(\hat E) + \left. \frac{\tx{d}f(E)}{\tx{d}E} \right|_{\hat E} (E(s)-\hat E(s)).
\end{align}
Thus, \eqref{eq:cv_force_max} is transformed into
\begin{align}
    F^\tx{lin}_{\gamma\tx{max}}(E)=
    \min\left\{\overline{F},\tx{y}_\tx{0}+\tx{y}_\tx{1}\sqrt{\frac{m}{2}}f^{\tx{lin}}(E,\hat E)\right\}
   \label{eq:cv_lin_Fmax}
\end{align}
and by using \eqref{eq:a_max},
\begin{align}
    a^\tx{lin}_\tx{max}(E)=\min\left\{\overline{\tx{a}},\frac{F^\tx{lin}_{\gamma\tx{max}}(E) -\tx{c}_\tx{a}E-F_\alpha}{\tx{m}} \right\}. \label{eq:cv_lin_a_max}
\end{align}
Also by having $F_{\gamma\tx{min}}^\tx{lin}(E)=0$ for the CV case study,
\begin{align}
    a^\tx{lin}_\tx{min}(E)=\max\left\{\underline{\tx{a}}, \frac{-\tx{c}_\tx{a}E+\underline{\tx{F}_\tx{brk}}-F_\alpha}{\tx{m}} \right \}. \label{eq:cv_lin_a_min}
\end{align}

The convex dynamic optimisation problem for the CV case study is now formulated as
{\allowdisplaybreaks
\begin{subequations} \label{eq:p4_cv}
\begin{align}
&\min_{j,F_\tx{brk}}\int_{0}^{s_\tx{H}} V_{\tx{CV}}(\cdot,\lambda_\tx{t},\hat E)\tx{d}s 
\label{eq:p4}\\
&\text{subject to:}\nonumber\\
&E'(s)=\tx{m}a(s) \label{eq:p4cv_E}\\
&a'(s)=j(s) \label{eq:p4cv_a}\\
&F(s)=\tx{m}a(s)+\tx{c}_\tx{a}E(s)-F_\tx{brk}(s)+F_\alpha(s)\\
&E(s)\in \frac{\tx{m}}{2} [v_\tx{min}^2(s), v_\tx{max}^2(s)] \label{eq:p4cv_Ebound}\\
&a(s)\in [a_\tx{min}^\tx{lin}(E),a_\tx{max}^\tx{lin}(E)]\label{eq:p4cv_abound}\\
& j(s)\in [\underline{\tx{j}}, \overline{\tx{j}}]\label{eq:p4cv_jbound}\\
& F_\tx{brk}(s) \in [\underline{\tx{F}_\tx{brk}},0]\label{eq:p4cv_Fbrkbound} \\
&E(0)=E_0, \quad a(0)=a_0 \label{eq:p4cv_int}
\end{align}%
\end{subequations}}%

After each SQP iteration, which occurs at each distance step forward, the trajectory about which that the problem is linearized is updated by moving towards the direction of the current optimal solution, i.e. 
\begin{align}
    \hat E^{(i+1)}(k) = \hat E^{(i)}(k) + \beta (E^{*(i)}(k) - \hat E^{(i)}(k)).
\end{align}
where $\beta$ is the step size that regulates the convergence rate.

For the EV case study, \eqref{eq:ev_force_min} is transformed into
\begin{align}
    F^\tx{lin}_{\gamma\tx{min}}(E)= \max\left\{\underline{F},f^{\tx{lin}}(E,\hat E)\right\}
   \label{eq:ev_lin_Fmin}
\end{align}
using the linearized function, $f^{\tx{lin}}(E,\hat E)$. Therefore, by using \eqref{eq:a_max}
\begin{align}
    a^\tx{lin}_\tx{min}(E)=\max\left\{\overline{\tx{a}},\frac{F^\tx{lin}_{\gamma\tx{min}}(E) -\tx{c}_\tx{a}E-F_\alpha}{\tx{m}} \right\}. \label{eq:ev_lin_a_min}
\end{align}

Note that the maximum traction force limit for EV is approximated by \eqref{eq:cv_force_max}. Accordingly, the maximum linearized acceleration is calculated by \eqref{eq:cv_lin_a_max}.

The convex dynamic optimisation problem for the EV case study is formulated as
{\allowdisplaybreaks
\begin{subequations} \label{eq:p4_ev}
\begin{align}
&\min_{j,F_\tx{brk}} \int_{0}^{s_\tx{H}} V_{\tx{EV}}(\cdot,\lambda_\tx{t},\hat E)\tx{d}s\\
&\text{subject to: \eqref{eq:p4cv_E}-\eqref{eq:p4cv_int}.}
\end{align}%
\end{subequations}}%

\section*{Acknowledgment}
This work has been financed by the Swedish Energy Agency (project number: 32226312). The authors would also like to
acknowledge Martin Sivertsson from Volvo Cars, Mikael Askerdal from Volvo Truck, and Henrik Sv{\"a}rd and Karl Redbrandt from Scania for the support and helpful discussions during the project.

\bibliographystyle{IEEEtran}
\bibliography{bibliography_ah}

\begin{thebibliography}{10}
\providecommand{\url}[1]{#1}
\csname url@samestyle\endcsname
\providecommand{\newblock}{\relax}
\providecommand{\bibinfo}[2]{#2}
\providecommand{\BIBentrySTDinterwordspacing}{\spaceskip=0pt\relax}
\providecommand{\BIBentryALTinterwordstretchfactor}{4}
\providecommand{\BIBentryALTinterwordspacing}{\spaceskip=\fontdimen2\font plus
\BIBentryALTinterwordstretchfactor\fontdimen3\font minus
  \fontdimen4\font\relax}
\providecommand{\BIBforeignlanguage}[2]{{%
\expandafter\ifx\csname l@#1\endcsname\relax
\typeout{** WARNING: IEEEtran.bst: No hyphenation pattern has been}%
\typeout{** loaded for the language `#1'. Using the pattern for}%
\typeout{** the default language instead.}%
\else
\language=\csname l@#1\endcsname
\fi
#2}}
\providecommand{\BIBdecl}{\relax}
\BIBdecl

\bibitem{oecditf19}
I.~T. Forum, \emph{ITF Transport Outlook 2019}.\hskip 1em plus 0.5em minus
  0.4em\relax OECD Publishing/ITF, 2019.

\bibitem{oecdiea18}
I.~E. Agency, ``Co2 emissions from fuel combustion 2018,'' p. 515, 2018.

\bibitem{kamal11}
M.~A.~S. Kamal, M.~Mukai, J.~Murata, and T.~Kawabe, ``Ecological vehicle
  control on roads with up-down slopes,'' \emph{IEEE Transactions on
  Intelligent Transportation Systems}, vol.~12, no.~3, pp. 783--794, 2011.

\bibitem{kamal12}
------, ``Model predictive control of vehicles on urban roads for improved fuel
  economy,'' \emph{IEEE Transactions on control systems technology}, vol.~21,
  no.~3, pp. 831--841, 2012.

\bibitem{vajedi15}
M.~Vajedi and N.~L. Azad, ``Ecological adaptive cruise controller for plug-in
  hybrid electric vehicles using nonlinear model predictive control,''
  \emph{IEEE Transactions on Intelligent Transportation Systems}, vol.~17,
  no.~1, pp. 113--122, 2015.

\bibitem{luo15}
Y.~Luo, T.~Chen, S.~Zhang, and K.~Li, ``Intelligent hybrid electric vehicle acc
  with coordinated control of tracking ability, fuel economy, and ride
  comfort,'' \emph{IEEE Transactions on Intelligent Transportation Systems},
  vol.~16, no.~4, pp. 2303--2308, 2015.

\bibitem{barkenbus2010eco}
J.~N. Barkenbus, ``Eco-driving: An overlooked climate change initiative,''
  \emph{Energy Policy}, vol.~38, no.~2, pp. 762--769, 2010.

\bibitem{bellman57}
R.~Bellman, \emph{{Dynamic Programming}}.\hskip 1em plus 0.5em minus
  0.4em\relax New Jersey: Princeton Univ Pr, 1957.

\bibitem{hellstrom09}
E.~Hellstr\"om, M.~Ivarsson, J.~\r{A}slund, and L.~Nielsen, ``Look-ahead
  control for heavy trucks to minimize trip time and fuel consumption,''
  \emph{Control Engineering Practice}, vol.~17, no.~2, pp. 245--254, 2009.

\bibitem{hellstrom10}
E.~Hellstr\"om, J.~\r{A}slund, and L.~Nielsen, ``Design of an efficient
  algorithm for fuel-optimal look-ahead control,'' \emph{Control Engineering
  Practice}, vol.~18, no.~11, pp. 1318--1327, 2010.

\bibitem{dib11}
W.~Dib, L.~Serrao, and A.~Sciarretta, ``Optimal control to minimize trip time
  and energy consumption in electric vehicles,'' in \emph{2011 IEEE Vehicle
  Power and Propulsion Conference}.\hskip 1em plus 0.5em minus 0.4em\relax
  IEEE, 2011, pp. 1--8.

\bibitem{heppeler16a}
G.~Heppeler, M.~Sonntag, U.~Wohlhaupter, and O.~Sawodny, ``Predictive planning
  of optimal velocity and state of charge trajectories for hybrid electric
  vehicles,'' \emph{Control Engineering Practice}, vol.~61, pp. 229--243, 2016.

\bibitem{wahl13}
H.-G. Wahl, K.-L. Bauer, F.~Gauterin, and M.~Holz{\"a}pfel, ``A real-time
  capable enhanced dynamic programming approach for predictive optimal cruise
  control in hybrid electric vehicles,'' in \emph{16th International IEEE
  Conference on Intelligent Transportation Systems (ITSC 2013)}.\hskip 1em plus
  0.5em minus 0.4em\relax IEEE, 2013, pp. 1662--1667.

\bibitem{buhler13}
L.~B\"uhler, ``Fuel-efficient platooning of heavy duty vehicles through road
  topography preview information,'' Master's thesis, KTH, Stockholm, Sweden,
  2013.

\bibitem{themann15}
P.~Themann, A.~Zlocki, and L.~Eckstein, \emph{Energieeffiziente
  Fahrzeugl{\"a}ngsf{\"u}hrung durch V2X-Kommunikation}.\hskip 1em plus 0.5em
  minus 0.4em\relax Wiesbaden: Springer Fachmedien Wiesbaden, 2015, pp. 27--33.

\bibitem{boyd04}
S.~Boyd and L.~Vandenberghe, \emph{Convex Optimization}.\hskip 1em plus 0.5em
  minus 0.4em\relax Cambridge University Press, 2004.

\bibitem{hellstrom10b}
E.~Hellstr{\"o}m, J.~\r{A}slund, and L.~Nielsen, ``Management of kinetic and
  electric energy in heavy trucks,'' \emph{SAE International Journal of
  Engines}, vol.~3, no.~1, pp. 1152--1163, 2010.

\bibitem{keulen10}
T.~van Keulen, B.~de~Jager, D.~Foster, and M.~Steinbuch, ``Velocity trajectory
  optimization in hybrid electric trucks,'' in \emph{American Control
  Conference}, Marriott Waterfront, Baltimore, MD, USA, 2010, pp. 5074--5079.

\bibitem{keulen11}
T.~van Keulen, B.~de~Jager, and M.~Steinbuch, ``Optimal trajectories for
  vehicles with energy recovery options,'' in \emph{IFAC World Congress},
  Milan, Italy, 2011, pp. 3831--3836.

\bibitem{pontryagin62}
L.~S. Pontryagin, V.~G. Boltyanskii, R.~V. Gamkrelidze, and E.~F. Mishchenko,
  \emph{The Mathematical Theory of Optimal Processes}.\hskip 1em plus 0.5em
  minus 0.4em\relax Interscience Publishers, 1962.

\bibitem{held18}
M.~Held, O.~Fl{\"a}rdh, and J.~M{\aa}rtensson, ``Optimal speed control of a
  heavy-duty vehicle in urban driving,'' \emph{IEEE Transactions on Intelligent
  Transportation Systems}, vol.~20, no.~4, pp. 1562--1573, 2018.

\bibitem{Keulen14}
T.~van Keulen, J.~Gillot, B.~de~Jager, and M.~Steinbuch, ``Solution for state
  constrained optimal control problems applied to power split control for
  hybrid vehicles,'' \emph{Automatica}, vol.~50, no.~1, pp. 187--192, 2014.

\bibitem{murgovski16}
N.~Murgovski, B.~Egardt, and M.~Nilsson, ``Cooperative energy management of
  automated vehicles,'' \emph{Control Engineering Practice}, vol.~57, pp.
  84--98, 2016.

\bibitem{johannesson15a}
L.~Johannesson, N.~Murgovski, E.~Jonasson, J.~Hellgren, and B.~Egardt,
  ``Predictive energy management of hybrid long-haul trucks,'' \emph{Control
  Engineering Practice}, vol.~41, pp. 83--97, 2015.

\bibitem{johannesson15b}
L.~Johannesson, M.~Nilsson, and N.~Murgovski, ``Look-ahead vehicle energy
  management with traffic predictions,'' in \emph{IFAC Workshop on Engine and
  Powertrain Control, Simulation and Modeling (E-COSM)}, vol.~48, Columbus,
  Ohio, USA, 2015, pp. 244--251.

\bibitem{hovgard18}
M.~Hovgard, O.~Jonsson, N.~Murgovski, M.~Sanfridson, and J.~Fredriksson,
  ``Cooperative energy management of electrified vehicles on hilly roads,''
  \emph{Control Engineering Practice}, vol.~73, pp. 66--78, 2018.

\bibitem{uebel17}
S.~Uebel, N.~Murgovski, C.~Tempelhahn, and B.~B{\"a}ker, ``Optimal energy
  management and velocity control of hybrid electric vehicles,'' \emph{IEEE
  Transactions on Vehicular Technology}, vol.~67, no.~1, pp. 327--337, 2017.

\bibitem{guo18}
L.~Guo, H.~Chen, Q.~Liu, and B.~Gao, ``A computationally efficient and
  hierarchical control strategy for velocity optimization of on-road
  vehicles,'' \emph{IEEE Transactions on Systems, Man, and Cybernetics:
  Systems}, vol.~49, no.~1, pp. 31--41, 2018.

\bibitem{turri16}
V.~Turri, B.~Besselink, and K.~H. Johansson, ``Cooperative look-ahead control
  for fuel-efficient and safe heavy-duty vehicle platooning,'' \emph{IEEE
  Transactions on Control Systems Technology}, vol.~25, no.~1, pp. 12--28,
  2016.

\bibitem{guo16}
L.~Guo, B.~Gao, Y.~Gao, and H.~Chen, ``Optimal energy management for hevs in
  eco-driving applications using bi-level mpc,'' \emph{IEEE Transactions on
  Intelligent Transportation Systems}, vol.~18, no.~8, pp. 2153--2162, 2016.

\bibitem{stroe17}
N.~Stroe, S.~Olaru, G.~Colin, K.~Ben-Cherif, and Y.~Chamaillard, ``A two-layer
  predictive control for hybrid electric vehicles energy management,''
  \emph{IFAC-PapersOnLine}, vol.~50, no.~1, pp. 10\,058--10\,064, 2017.

\bibitem{uebel19}
S.~Uebel, N.~Murgovski, B.~Baker, and J.~Sjoberg, ``A 2-level mpc for energy
  management including velocity control of hybrid electric vehicle,''
  \emph{IEEE Transactions on Vehicular Technology}, 2019.

\bibitem{hanson17}
B.~B. Hanson and T.~E. Hanson, ``Systems and methods for multi-mode unmanned
  vehicle mission planning and control,'' Jun.~6 2017, uS Patent 9,669,904.

\bibitem{hamednia18}
A.~Hamednia, N.~Murgovski, and J.~Fredriksson, ``Predictive velocity control in
  a hilly terrain over a long look-ahead horizon,'' \emph{IFAC-PapersOnLine},
  vol.~51, no.~31, pp. 485--492, 2018.

\bibitem{diehl01}
M.~Diehl, ``Real-time optimization for large scale nonlinear processes,'' Ph.D.
  dissertation, University of Heidelberg, 2001.

\bibitem{thomas94}
M.~M. Thomas, J.~Kardos, and B.~Joseph, ``Shrinking horizon model predictive
  control applied to autoclave curing of composite laminate materials,'' in
  \emph{Proceedings of 1994 American Control Conference-ACC'94}, vol.~1.\hskip
  1em plus 0.5em minus 0.4em\relax IEEE, 1994, pp. 505--509.

\bibitem{murgovski13china}
N.~Murgovski, X.~Hu, L.~Johannesson, and B.~Egardt, ``Filtering driving cycles
  for assessment of electrified vehicles,'' in \emph{Workshop for new energy
  vehicle dynamic system and control technology}, Beijing, China, 2013.

\bibitem{lipp14}
T.~Lipp and S.~Boyd, ``Minimum-time speed optimization along a fixed path,''
  \emph{International Journal of Control}, vol.~87, no.~6, pp. 1297--1311,
  2014.

\bibitem{murgovski15ACC}
N.~Murgovski, L.~Johannesson, X.~Hu, B.~Egardt, and J.~Sj\"oberg, ``Convex
  relaxations in the optimal control of electrified vehicles,'' in
  \emph{American Control Conference}, Chicago, USA, 2015.

\bibitem{castro16}
R.~de~Castro, M.~Tanelli, R.~E. Ara\'ujo, and S.~M. Savaresi, ``Minimum-time
  path-following for highly redundant electric vehicles,'' \emph{IEEE
  Transactions on Control Systems Technology}, vol.~24, no.~2, pp. 487--501,
  2016.

\bibitem{diehl05}
M.~Diehl, H.~G. Bock, and J.~P. Schl{\"o}der, ``A real-time iteration scheme
  for nonlinear optimization in optimal feedback control,'' \emph{SIAM Journal
  on control and optimization}, vol.~43, no.~5, pp. 1714--1736, 2005.

\bibitem{wahlstrom11}
J.~Wahlstr\"om and L.~Eriksson, ``Modelling diesel engines with a
  variable-geometry turbocharger and exhaust gas recirculation by optimization
  of model parameters for capturing non-linear system dynamics,'' \emph{Sage
  journals}, vol. 225, no.~7, pp. 960--986, 2012.

\bibitem{eriksson16}
L.~Eriksson, A.~Larsson, A.~Thomasson, and S.~C. Ab, ``Heavy duty truck on open
  road--the aac2016 benchmark,'' in \emph{IFAC Symposium on Advances in
  Automotive Control}, 2016.

\end{thebibliography}

\end{document}